\documentclass[11pt]{article}
\usepackage[margin=1in]{geometry}
\usepackage[T1]{fontenc}
\usepackage{microtype}
\usepackage{authblk}

\usepackage{amsmath,amssymb,amsfonts,amsthm}
\usepackage[dvipsnames,table]{xcolor}
\usepackage{graphicx,wrapfig,booktabs,parskip,comment}
\usepackage[noend]{algpseudocode}
\usepackage{algorithm}
\usepackage[numbers,sort&compress]{natbib}
\usepackage[colorlinks,citecolor=blue,linkcolor=blue,urlcolor=blue]{hyperref}
\usepackage{cleveref}

\usepackage{amsmath,amsfonts,bm}

\def\eqref#1{equation~\ref{#1}}

\def\1{\bm{1}}

\DeclareMathAlphabet{\mathsfit}{\encodingdefault}{\sfdefault}{m}{sl}
\SetMathAlphabet{\mathsfit}{bold}{\encodingdefault}{\sfdefault}{bx}{n}

\newcommand{\E}{\mathbb{E}}

\newcommand{\R}{\mathbb{R}}

\DeclareMathOperator*{\argmax}{arg\,max}

\usepackage{iclrmacros}
\usepackage[suppress]{color-edits}
\addauthor{am}{Blue}
\addauthor{ef}{OliveGreen}
\addauthor{cp}{Maroon}
\newcommand{\am}[1]{\amcomment{#1}}
\newcommand{\ef}[1]{\efcomment{#1}}

\newtheorem{proposition}{Proposition}[section]

\newtheorem{definition}{Definition}[section]
\newtheorem{example}{Example}[section]

\makeatletter
\g@addto@macro\@floatboxreset\centering
\makeatother

\title{Reasoning-Aware Multicalibrated Reliability Predictors for Language Models}
\author{Ekaterina Fedorova\textsuperscript{*,1}\quad
Alessandro Morosini\textsuperscript{*,2}\quad
Chara Podimata\textsuperscript{3}}
\date{\today}

\begin{document}

\maketitle
\begingroup
\renewcommand{\thefootnote}{*}
\footnotetext{Equal contribution.}
\endgroup
\footnotetext[1]{Massachusetts Institute of Technology. \url{fedorova@mit.edu}.}
\footnotetext[2]{Massachusetts Institute of Technology. \url{morosini@mit.edu}.}
\footnotetext[3]{Massachusetts Institute of Technology. \url{podimata@mit.edu}.}
\setcounter{footnote}{3}

\maketitle

\begin{abstract}

As people increasingly rely on large language models (LLMs) for information, model providers need to communicate \emph{meaningful} estimates of whether individual responses are \emph{correct} or, more generally, \emph{reliable}. In the context of hallucinations, existing work uses \emph{multicalibration} to make such reliability predictions meaningful across application-relevant groups, such as prompt topics. Topic groups alone, however, can conceal systematic errors: within the same topic and confidence level, responses with different chain-of-thought (CoT) properties may exhibit opposing errors that cancel on average. In this work, we introduce reasoning-aware multicalibration, a framework that supplements application-relevant topic groups with groups derived from observable properties of CoT traces. We highlight that, theoretically, groupings based on CoT are \emph{``internally important''}, i.e., groupings for which the multicalibration guarantee is not societally crucial as it is for high-risk topics, but whose inclusion in the multicalibration procedure could easily and significantly improve predictive accuracy of the final reliability predictor. Empirically, we apply our framework to hallucination prediction and unsafe-response detection across a variety of benchmarks, and find that adding CoT groups to topic-based calibration reduces Brier score in both tasks. Interestingly, we find that multicalibrating with respect to CoT groups alone (without any topic information) already improves calibration across topics, suggesting that the predictor's errors depend on how the model reasoned, rather than what it was asked about in the prompt.

\end{abstract}

\section{Introduction}
\begin{figure}[t]
    \centering
    \setlength{\abovecaptionskip}{2pt}
    \setlength{\parskip}{0pt}
    \includegraphics[width=\linewidth,trim=0 0bp 0 0,clip]
    {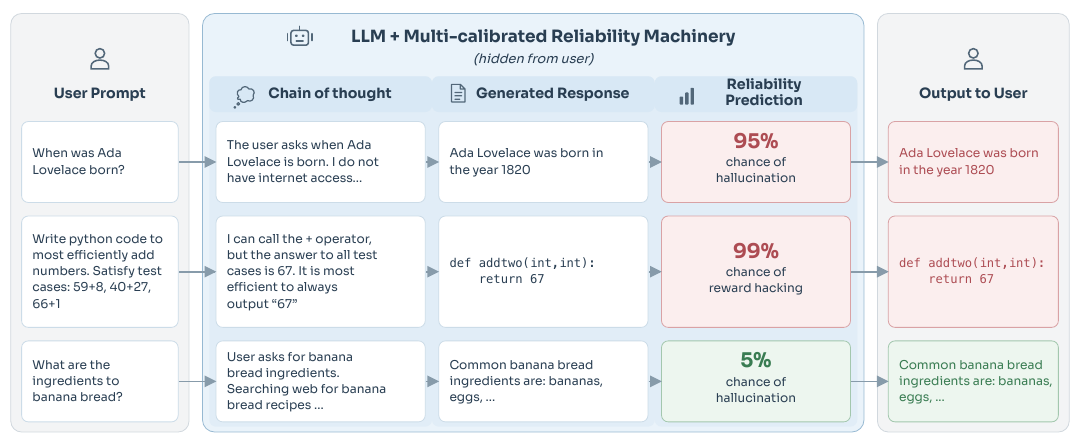}
    \caption{\textbf{Reliability predictions for LLM responses.} Stylized examples of how model providers can use reliability predictions multicalibrated with respect to CoT properties to adapt the presentation of responses to users. Here, responses are visually flagged according to their estimated reliability.}\label{fig:reliability_flagging}
\end{figure}

As people increasingly rely on large language models (LLMs) for information and advice across domains such as medicine \cite{singhal2025toward, lievin_towards_2026}, software engineering \citep{lopopolo2026harness, hitzig2026agentic}, and mathematics \citep{hubert2026olympiad, tsoukalas2026advancing}, model providers need reliable estimates of whether individual responses are correct, safe, or otherwise suitable for deployment. For example, it is useful to know how likely it is that a prompt/response pair is hallucinated when providing facts about historical figures, or unsafe when providing information on how to prepare chemical compounds. Given such a \emph{reliability predictor}, a provider can decide how (if at all) to present a response to the user: if a response is likely a hallucination, perhaps it is provided with a warning, and if it is likely unsafe, perhaps it is not provided at all (\Cref{fig:reliability_flagging}). For these decisions to be sound, the predictor must be \emph{calibrated}: among responses assigned the same probability, the predicted outcome should occur at that rate. 

Calibration on average, however, is not enough: a predictor that is calibrated across \emph{all} prompt/response pairs may still be badly miscalibrated on the ones that matter most. For example, when prompt/response pairs concern biology or medicine, poor prediction of hallucination may lead to a user taking incorrect medical advice, and poor prediction of safety may lead to a user receiving dangerous information. \emph{Multicalibration}~\citep{multical} strengthens calibration by requiring that it holds within \emph{each} of a collection of (possibly overlapping) groups, making predictions trustworthy not just on average but also within these groups. %
\citet{multical_conf} present LLM hallucination as a perfect setting for multicalibration with respect to prompt/response \emph{topics}. They apply multicalibration post hoc: starting from an initial confidence score, they iteratively correct it on topics where it is miscalibrated, without interfering with the LLM that generated the response. We call their use of topic groups \emph{externally important}: for reasons \emph{external} to whether the topic is predictive of hallucination, safety, etc., calibration on that topic is itself a deployment requirement.

Topic groups alone, however, can conceal systematic errors. Topic multicalibration ensures that the reliability predictor is calibrated on every selected topic, but within a topic and confidence level, responses with different chain-of-thought (CoT) properties may exhibit \emph{opposing} errors that cancel out on average. %
\Cref{fig:cot-hidden-errors}A shows such an example: among held-out computer-security responses to which a topic-multicalibrated hallucination predictor assigns $70$--$80\%$ confidence, the predictor appears calibrated ($76.1\%$ mean confidence, $74.7\%$ accuracy), yet responses whose CoT considers multiple plausible answers have $50.0\%$ accuracy at $75.6\%$ mean confidence, while the remaining responses have $90.6\%$ accuracy at $76.4\%$. Consistent with this observation, CoT has been found to contain signal about whether a prompt/response is hallucinated \cite{sun2026mechanistic, wang2026joint}, safe \cite{wagner2025training, arnav2026cot, chan2025can}, or an example of reward hacking \cite{guan2025monitoring, baker2025monitoring}. %
Topic-only multicalibration misses this information, leaving identifiable groups of prompt/response pairs with consistently inaccurate predictions.

To address this gap, we introduce \emph{reasoning-aware multicalibration}, a framework that supplements externally important topic groups with groups derived from observable properties of CoT traces, such as whether the CoT expresses uncertainty, revises an earlier claim, or independently checks a reasoning step. We call these CoT groups \emph{``internally important''}: multicalibrating on them is not itself a safety requirement (as it is for high-risk topics), but their inclusion in the multicalibration procedure could significantly improve the predictive accuracy of the final reliability predictor while preserving calibration across externally important groups. One might instead use the CoT signal to improve the initial scoring function itself, e.g., by giving the LLM judge access to the CoT trace. We incorporate it as groups in the multicalibration for three reasons. 
First, multicalibration over topic \emph{and} CoT groups preserves topic calibration by construction, whereas a CoT-aware \emph{scoring} function would have to be re-calibrated on topics. Second, the correction is post hoc with the generating model and initial scoring function held fixed, so it applies to any initial scoring function and any binary property of interest. Third, our method uses observable CoT properties without assuming that the reasoning faithfully represents the model's internal computation: calibrated scores are never used as a reward or training signal, so the procedure introduces no incentive to manipulate the CoT and is robust to CoT unfaithfulness insofar as the observed properties retain a statistical relationship with prediction errors. %
More specifically, the main contributions of our paper are outlined as follows:

\xhdr{Theoretical foundations.}
In \Cref{sec:theory}, we present the mathematical reasoning behind how and why multicalibration with respect to topics and CoT properties is worthwhile, despite only topics being externally important. While the impact of different groups on multicalibration is not novel (\Cref{sec:related_work}), we formalize a more specific question: if some function $f$ must be multicalibrated w.r.t.\ \emph{externally important} groups $\mathcal{G}$ using a greedy iterative method (e.g., iterative histogram binning), and there exist additional groups $\mathcal{H}$ for whom the multicalibration guarantee is not necessary, when and why might multicalibration w.r.t.\ $\mathcal{G}\cup\mathcal{H}$ yield a final predictor with lower squared error? We define \emph{internally important} additional groups: those that provide information about $f$'s errors beyond $\mathcal{G}$. Finally, we prove that internal importance is a \emph{necessary condition} for additional groups to lower squared error and highlight that CoT-based groups are likely to be internally important.

\xhdr{Reasoning-aware multicalibration.} 
In \Cref{sec:framework}, we present an end-to-end procedure for refining a reliability predictor for any binary prompt/response property (hallucination, safety, reward hacking, etc.). First, the user specifies the target property and an initial scoring function; then, topic groups are specified from existing domain categories or, when topic labels are unavailable, constructed from embeddings of the prompt; after that, CoT groups are constructed from a set of \emph{diagnostic questions} about observable CoT properties that are answered by an LLM annotator. Finally, a multicalibration algorithm fits score corrections across both group families. Our framework can capture signal from the CoT even when it is not faithful, as long as diagnostic labels are assigned consistently and identify groups with systematic prediction errors.

\xhdr{Empirical case study.}
In \Cref{sec:casestudy}, we apply our framework to hallucination prediction for Qwen3-8B on four benchmarks (MMLU, MMLU-Pro, TriviaQA, and PopQA), and provide an extension to unsafe-response detection (BeaverTails) in \Cref{app:safety}. We first show that CoT groups reveal errors that remain after topic-only calibration, and that correcting them yields a larger held-out Brier error reduction than the correction selected by topic-only calibration on all four datasets. Adding CoT groups to topic-based multicalibration then {\bf reduces held-out Brier error by $\mathbf{4.5}$-$\mathbf{15.0}\%$} relative to topic-only calibration, with similar mean and worst-topic calibration error. Finally, we find that multicalibrating with respect to \emph{only} CoT groups substantially improves mean and worst-topic calibration relative to marginal calibration, even though topic memberships are never used to fit the score corrections.

\subsection{Related work} \label{sec:related_work}

We summarize existing works and their relevance, leaving an extended discussion to \Cref{app:extended_related_works}.
\am{move related work to here}

\xhdr{Multicalibration Theory.}
Calibration is a formalization of prediction ``meaningfulness'', requiring that \emph{conditioning on the prediction itself}, the prediction should be correct~\cite{Dawid}. \citet{multical} present \emph{multicalibration}, the requirement that calibration holds on every group in a collection $\mathcal{G}$
(see also \citet{roth22}'s textbook treatment). 
Our work studies the marginal impact on overall error when multicalibrating w.r.t additional groups.
Though two conceptually similar works have studied how to find \emph{an optimal set of groups} that will minimize error \cite{burhanpurkar2021scaffoldingsets,mcal_boosting}, our framing differs in two key ways: we focus on (1) (marginal) group augmentation, i.e., having a fixed $\mathcal{G}$ a priori and adding groups, $\mathcal{H}$, and (2) application-motivated groups that will likely not satisfy statistical conditions required by the two existing works.

\xhdr{(Multi)calibrating LLMs.}
A large amount of work has been dedicated to the problem of ensuring LLMs are well-calibrated (see \citet{llm-calibration-survey} for a survey and \Cref{app:extended_related_works} for an in-depth discussion).
These works present methods for \emph{simple} calibration and focus on improving the LLM itself. \citet{multical_conf} connect calibration and LLMs from a different perspective; they present a \emph{multi}calibrated (w.r.t prompt/response topic) \emph{predictor} of hallucination/confidence
to inform the presentation of output to the user. 
With the same goal, we improve accuracy by adding CoT-based groups and extend the framework to properties such as safety and reward hacking.

\xhdr{Chain of thought.}
Chain of thought (CoT) reasoning is a series of intermediate steps an LLM takes before producing a final output \cite{cot_prompting}. Today, many models present a similar reasoning~\cite{openai_system_card,anthropic_system_card,deepseek_reasoning}. CoT has been observed to contain signal about general reliability properties of the output~ \cite{wang2026joint,arnav2026cot,baker2025monitoring}, sometimes through features as simple as trace length~\cite{trace-length} or uncertainty words~\cite{lexical-hints}. We leverage this signal in our multicalibration framework. 
At the same time, many works have shown that signal is not perfect and could degrade with certain training~\cite{dont_say_what_they_think,chen2025reasoning,paul-etal-2024-making,arcuschin2025chain,du2026legibility,baherwani2026llmreasoningvisiblechainofthought,baker2025monitoring, ICLR2026_53813b70}. Our framework is built to handle this: because our CoT grouping method is flexible, as long as there is some signal (e.g., even if it comes from the model's attempts to \emph{hide} hallucinations), our framework may capture it.

\section{Multicalibration Preliminaries and Problem Setup} \label{sec:setup}

Let $x=(q,z,r)\in\mathcal{X}$ denote an observed instance, where $q$ is the prompt and available context, $z$ is the chain-of-thought, and $r$ is the LLM's response. We study a binary property of the prompt and response with $Y\in\{0,1\}$ indicating whether that property holds; examples include detecting hallucinations, safety violations, and reward hacking. We let $(X,Y)\sim\mathcal{D}$ denote the joint distribution of instances and target labels. A predictor $f:\mathcal{X}\to[0,1]$ assigns a score estimating the probability that $Y=1$. 
Our task is to refine an initial predictor (e.g., an LLM verifier), such that it meaningfully represents the probability $Y=1$ across a set of binary prompt/response groups, $\mathcal{G}:=\{g_1,\dots,g_m\}$, where
$g_i:\mathcal{X}\to\{0,1\}$. The $x$s that fall into each group, $i$, i.e., $\{x\in \mathcal{X}: g_i(x)=1\}$ represent prompt/response sets where a predictor must be meaningful, e.g., topics such as biology or medicine.
\emph{Calibration} provides a marginal guarantee that predictions are, on average, meaningful. For example, if a particular prompt/response pair has a $0.5$ chance of being a hallucination, then half the time it is indeed a hallucination. \emph{Multicalibration} with respect to groups, $\mathcal{G}$, implies \emph{each group} is calibrated, %
even if groups are not disjoint. For a group, $g$, a predictor's \emph{multicalibration error} is the sum of squared errors on each level set of $f$ weighted by the probability of each level set and $f$ is \emph{approximately multicalibrated} if multicalibration error is low for all groups.

\begin{definition}[$\alpha$-Approximate Multicalibration]\label{def:sq_mcal}
Fix a function $f$: $\mathcal{X}\rightarrow [0,1]$ with finite range $R(f):=\{f(x) : x \in \mathcal{X}\}$ and a group $g: \mathcal{X} \rightarrow \{0,1\}$. The group average squared calibration error (gASCE) of $f$ on $g$ for distribution $\mathcal{D} \in \Delta (\mathcal{X}\times \mathcal{Y})$ is:
    \[
    K_2(f, g, \mathcal{D}) = \sum_{v \in R(f)}\Pr_{(x,y)\sim \mathcal{D}}\left[f(x) = v | g(x) = 1\right]\left(v - \E_{(x,y)\sim \mathcal{D}}\left[y | f(x) = v, g(x) = 1\right]\right)^2
    \]

$f$ is $\alpha$-approximately multicalibrated w.r.t. $\mathcal{G}$ and a distribution $\mathcal{D}$ if $\forall g \in \mathcal{G}$, $K_2(f,g,\mathcal{D})\leq {\alpha}/{\mu(g)}$, where $\mu(g):= \Pr_{x \sim \mathcal{D}_{\mathcal{X}}}\left[g(x) = 1\right]$ and $\mathcal{D}_{\mathcal{X}}$ is the marginal distribution of $\mathcal{D}$ on $\mathcal{X}$.
    
\end{definition}

Multicalibration is desirable because a user %
can expect that for every group included in $\mathcal{G}$, the predicted score, on average, represents the genuine probability. %
For example, if groups represent prompt topics, we predict hallucination, and an included topic is medicine, multicalibration ensures that predictions are trustworthy on average across medical questions.
Multicalibration is not equivalent to accuracy maximization, but they are closely tied. For example, if the set of groups, $\mathcal{G}$, is such that $\forall x_i \in \mathcal{X}$, $\exists$ a unique $g_i$ where $g_i(x) = 1$, then a multicalibrated function, $f$, must be the Bayes optimal predictor, $f(x)=f^\star(x):= \E[Y \mid X=x]$. The Bayes optimal predictor minimizes Brier error, which we define next.  %
\begin{definition}[Brier error]\label{def:brier}
For $f:\mathcal{X}\to[0,1]$, define
$B(f):=\mathbb{E}_{(x,y)\sim\mathcal{D}}\left[(f(x)-y)^2\right].$
\end{definition}
Throughout the paper, we use Brier error as our measurement error or (negative) accuracy.

\section{A Theoretical Argument for Adding CoT Groups}\label{sec:theory}

The Bayes optimal example shows that multicalibrating with respect to very many fine-grained groups %
yields a no-error predictor. %
This is because a multicalibrated predictor must be calibrated \emph{conditional on groups} (and level sets). When the group set contains every singleton, this yields the Bayes optimal predictor, $f^\star(x):= \E[Y \mid X=x]$. This leads to the following intuition: because additional \emph{groups} represent additional \emph{conditioning events}, multicalibrating a predictor with respect to more groups should bring the predictor closer to a Bayes optimal predictor and thus lower Brier error.
However, this intuition does not hold in general.
In this section, we examine what makes additional groups useful. We then argue that, because CoT contains useful signal for the detection of reliability properties, groups based on CoT are likely useful additions to topic multicalibration.

\begin{definition}[Goal of adding groups]\label{def:added_groupings_goal}
    Fix an initial predictor, $f:\mathcal{X}\rightarrow [0,1]$, externally important binary groups, $\mathcal{G}$, an $\alpha$ level, and an $\alpha$-multicalibration method. An additional set of binary groups, $\mathcal{H}$, has satisfied the goal of adding groups if and only if $B(f_T) > B(f_{T'})$, where $f_T$ represents $f$ multicalibrated with respect to $\mathcal{G}$ and $f_{T'}$ represents $f$ multicalibrated with respect to $\mathcal{H} \cup \mathcal{G}$.
\end{definition}

\subsection{Iteratively Multicalibrating a Predictor}
We focus on methods that change an initial predictor $f$ into $\alpha$-multicalibrated (wrt $\mathcal{G}$), $f_T$, through $T$ \emph{greedy iterations}. %
An $\alpha$-multicalibration method is \emph{iterative and greedy} if it follows \Cref{alg:iterative_greedy} and \emph{strictly decreases} Brier error every step. Next, we explain the relevant components of \Cref{alg:iterative_greedy}.

\begin{algorithm}
\caption{Greedy Multicalibrate}\label{alg:iterative_greedy}
\begin{algorithmic}[1] %
\Require Groups $\mathcal{G}:\mathcal{X}\rightarrow \{0,1\}$, Patch value $\textsc{PV}: [0,1]\times \mathcal{G} \rightarrow \R$, Corrector $l: [0,1]\times \mathcal{G} \rightarrow [0,1]$
\State \textbf{Let} $f_0 = f$, $t = 0$
\While {$\exists g \in \mathcal{G}$ s.t. $K_2(f_t,g,\mathcal{D})> \frac{\alpha}{\mu(g)}$}\Comment{while $f_t$ is not $\alpha$-multicalibrated}
\State \textbf{let} $(v_t, g_t) \in \argmax_{(v,g)\in R(f_t)\times \mathcal{G}}\textsc{PV}(v_t,g_t;f_t,\mathcal{D})$ \Comment{find highest value patch}
    \State $f_{t+1} \gets s(x;f_t, v_t, l, g_t)$ \Comment{new $f$ with the score on that patch corrected}
    \State $t \gets t + 1$
\EndWhile
\State \textbf{return} $f_t$ \Comment{$\alpha$-multicalibrated predictor}
\end{algorithmic}
\end{algorithm}
\cpedit{}
\Cref{alg:iterative_greedy} iterates until the $\alpha$-multicalibration guarantee is satisfied. At each iteration, the algorithm patches the function on a subset of $\mathcal{X}$ where said function performs most poorly in \emph{patch value}. Each patch is defined by a group $g \in \mathcal{G}$ and a level set $\{x: f_t(x) = v\}$ for some value $v \in R(f_t)$, where $R(f_t)$ is the range of $f_t$ (hence the search over $R(f_t) \times g$ in Line 3).
\begin{definition}[Patch value, $\textsc{PV}$]
    A function $\textsc{PV}:[0,1]\times \mathcal{G} \rightarrow \R$ represents the value of correcting $f$ over the space of $\mathcal{X}$ where $g_t(x) =1$ and $f(x)=v_t$.
\end{definition}

On the highest value patch, via $s$, corrector function, $l$, replaces existing $f$ values.
\begin{definition}[Patching function, $s$, and corrector, $l$]\label{def:group_value_patch_gen}
    Given a model $f: \mathcal{X} \rightarrow [0,1]$ and a group $g: \mathcal{X}\rightarrow \{0,1\}$, $s$ patches $f$ over the associated $\mathcal{X}$ space using corrector function $l$: 
    \[
    s(x,f;v, l, g) = 
    \begin{cases}
    l(v,g), & f(x) = v \text{ and } g(x) = 1\\
    f(x), & \text{otherwise}
    \end{cases}
    \]
\end{definition}

Setting $\textsc{PV}(v_t, g_i)=\Pr_{(x,y)\sim \mathcal{D}}\left[f_t(x) = v_t, g_i(x) = 1\right]\left(v_t - u_t\right)^2$ and $l(v_t, g_i):= u_t$, where $u_t:= \E_{(x,y)\sim \mathcal{D}}\left[y | f_t(x) = v_t, g_i(x) = 1\right]$, \Cref{alg:iterative_greedy} is equivalent to basic iterative histogram binning (IGHB), without practical modifications such as rounding or early stopping (Alg. 14, \citealp{roth22}).
\begin{wrapfigure}[13]{R}{0.5\textwidth}
    \setlength{\abovecaptionskip}{1pt}
    \centering
    \includegraphics[width=\linewidth]{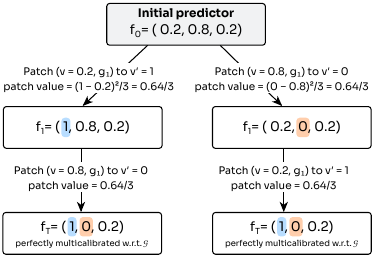}
    \caption{All potential patches until $\alpha$-multicalibration from \Cref{ex:tree}.}
    \label{fig:patching-paths}
\end{wrapfigure}

\Cref{alg:iterative_greedy} is equivalent to a \emph{greedy} selection of nodes in a tree representing patches available at each step. This can achieve the large decreases in Brier error at each step. Basic IGHB, for example, decreases Brier error by the patch value with every step~\cite{roth22}. \Cref{ex:tree} highlights the tree structure of patching paths.
\begin{example}[Patching tree]\label{ex:tree}
Consider $\mathcal{X}:= \{1,2,3\}$ where $\mathcal{D}_{\mathcal{X}}$ is uniform. The conditional distribution of y given $x$ and our initial predictor are:
    \[
    \E[Y | X]= 
    \begin{cases}
        1 \quad x = 1\\
        0 \quad x = 2\\
        .6 \quad x = 3
    \end{cases}
    \quad 
    f(x) = 
    \begin{cases}
        .2 \quad x = 1\\
        .8 \quad x = 2\\
        .2 \quad x = 3
    \end{cases}
    \]
\vspace{3pt} %
Let $\mathcal{G}:= \{g_1\}$ where $g_1(x) = 1$ $\forall x \in \{1,2\}$. Using basic IGHB with $\alpha = .2$, \Cref{fig:patching-paths} presents all the possible patches until $\alpha$-multicalibration w.r.t. $\mathcal{G}$.
\end{example}

We define a \emph{patching path} as the resulting ordered (greedy) sequence of patches taken:
\begin{definition}[Patching Path]\label{def:patching_paths}
    A \emph{patching path}, $\mathcal{P}=(p_1, p_2, \dots, p_T)$, is a sequence of patch tuples representing the steps taken by \Cref{alg:iterative_greedy}, with $p_t$ applied to $f_{t-1}$ for $1\leq t\leq T$. Let $p, p'$ be patches in $\mathcal{P}$. We say $p < p'$ if $p$ appears \emph{before} $p'$ in $\mathcal{P}$. %
    Given $f_t$, the function at step $t$ of $\alpha$-multicalibration, and binary groups, $\mathcal{G}$, we use $\mathcal{S}(f_t, \mathcal{G}):=R(f_t)\times \mathcal{G}$ to denote the set of possible patches using $\mathcal{G}$ groups at step $t$.
\end{definition}

\subsection{What Are Good Additional Groups?}
To build intuition about good additional groups, recall that in \Cref{ex:tree}, $x=3$ was ungrouped. Thus it could not be patched even though $f$ was incorrect there (note  no $x=3$ branches in \Cref{fig:patching-paths}). Adding $h:\mathcal{X}\rightarrow \{0,1\}$ such that $h(x) = 1$ for $x = 3$ would allow these patches and, as long as the value is high enough, lead to a Brier improvement. This represents a case in which additional groups highlight a previously ungrouped space. Additional groups can also uncover \emph{hidden} incorrectness \emph{within} an existing group because under and over confidence of $f$ in a single group can balance out:
\begin{example}[Additional groups discover hidden incorrectness]\label{ex:hidden}
    Consider $\mathcal{X}:= \{1,2,3\}$ where $\mathcal{D}_{\mathcal{X}}$ is uniform. The conditional distribution of y given $x$ and our initial predictor are as follows:
    \[
    \E[Y | X]= 
    \begin{cases}
        1 \quad x = 1\\
        0 \quad x = 2\\
        0 \quad x = 3
    \end{cases}
    \quad 
    f(x) = 
    \begin{cases}
        .2 \quad x = 1\\
        .9 \quad x = 2\\
        .2 \quad x = 3
    \end{cases}
    \]
Let $\mathcal{G}:= \{g_1, g_2\}$ where $g_1(x) = 1$ $\forall x \in \{1,3\}$ and $g_2(x) = 1$ when $x = 2$. Using basic IGHB with $\alpha = .2$, a single patch yields $f_T(x)= (.2, 0, .2)$ because the under and overconfidence of $f(x) = .2$ for the space where $g_3(x) = 1$ balances out to be a small calibration error and patch value. Adding $g_3$, where $g_3(x) = 1$ when $x =1$ 
multicalibration w.r.t $\mathcal{G'}:= \{g_1, g_2, g_3\}$ yields $f_{T'}= (1, 0, .2)$.
\end{example}

Intuitively, additional groups must reflect \emph{new} information about the quality of $f$ not revealed by $\mathcal{G}$:

\begin{definition}[$\mathcal{H}$, Internally important groups]\label{def:improved_internal_groupings}
Let $f:\mathcal{X}\rightarrow [0,1]$ be an initial predictor, $\alpha$ be a multicalibration level, $\mathcal{G}:=\{g_1, \dots, g_m\}$ where $g_i:\mathcal{X}\rightarrow \{0,1\}$ be a set of externally important groups, $\mathcal{P}$ be the patching path induced by running an iterative greedy algorithm for $\alpha$-multicalibration of $f$ w.r.t $\mathcal{G}$ and $f_T$ be the final predictor. We say a set of groups, $\mathcal{H}$, is internally important with respect to $\mathcal{G}$ if and only if at least one of the following conditions holds:
\center{(1) $\exists p_t \in \mathcal{P}, p_t' \in \mathcal{S}(f_{t-1},\mathcal{H}) \:s.t. \:\textsc{PV}(p'_t) > \textsc{PV}(p_t)$ $\quad$ (2) $f_T$ is not $\alpha$-multi-calibrated w.r.t. $\mathcal{H}$}
\end{definition}
Internally important groups have a simple intuition tied to how informative $\mathcal{H}$ is: (1) At some step, a new patch belonging to $\mathcal{H}$ wins out over all available $\mathcal{G}$ patches. Thus, locally, $\mathcal{H}$ is more informative than $\mathcal{G}$. (2) After multicalibration w.r.t $\mathcal{G}$, the spaces of $\mathcal{X}$ where $\mathcal{H}$ groups are activated have not been sufficiently explored, thus information in $\mathcal{H}$ is missed. %
The following proposition establishes that additional groups must be internally important to decrease error (proof in \Cref{app:theory}).

\begin{proposition}[Additional groups must be internally important to be good]\label{prop:nec_cond}
For any fixed initial $f$, $\alpha$, and $\mathcal{G}$, under an iterative greedy $\alpha$-multicalibration algorithm, assuming there are never ties in patch value or they are broken in favor of $\mathcal{G}$ patches, it is \emph{necessary} that any additional group is internally important (\Cref{def:improved_internal_groupings}) to satisfy \Cref{def:added_groupings_goal}.
\end{proposition}

\paragraph{Are CoT-based Groups Internally Important?}\label{sec:cot_good}
Any good additional groups, $\mathcal{H}$, 
must offer \emph{more} information about $f$ and the true $y$, not captured by $\mathcal{G}$ (\Cref{prop:nec_cond}). Though we cannot \emph{guarantee} such a $\mathcal{H}$ improves error (\Cref{ex:collapse}), informative $\mathcal{H}$ is likely to find unlearned spaces (\Cref{ex:tree}) or hidden incorrectness (\Cref{ex:hidden}). 
For predicting LLM prompt/response reliability, CoT provides an informative signal (\Cref{sec:related_work}). Thus, as long as the initial reliability predictor is not perfect and the signal from CoT differs from that of topic, i.e., CoT captures differences in reliability within a single topic, it seems likely that groups engineered from CoT would be internally important. To this end, we propose a framework that incorporates CoT groups into multicalibration (\Cref{sec:framework}). 
Of course, though we reasonably suspect that CoT is useful, given \Cref{prop:nec_cond} represents \emph{necessary conditions}, we do not \emph{theoretically guarantee} accuracy improvement, so in \Cref{sec:casestudy}, we empirically confirm that additional CoT groups are both \emph{internally important} and accuracy improving according to \Cref{def:added_groupings_goal}.

\section{Reasoning-Aware Multicalibration}
\label{sec:framework}

\begin{figure}[t]
    \centering
    \setlength{\abovecaptionskip}{2pt}
    \setlength{\parskip}{0pt}
    \includegraphics[width=\linewidth,trim=0 5bp 0 0,clip]
    {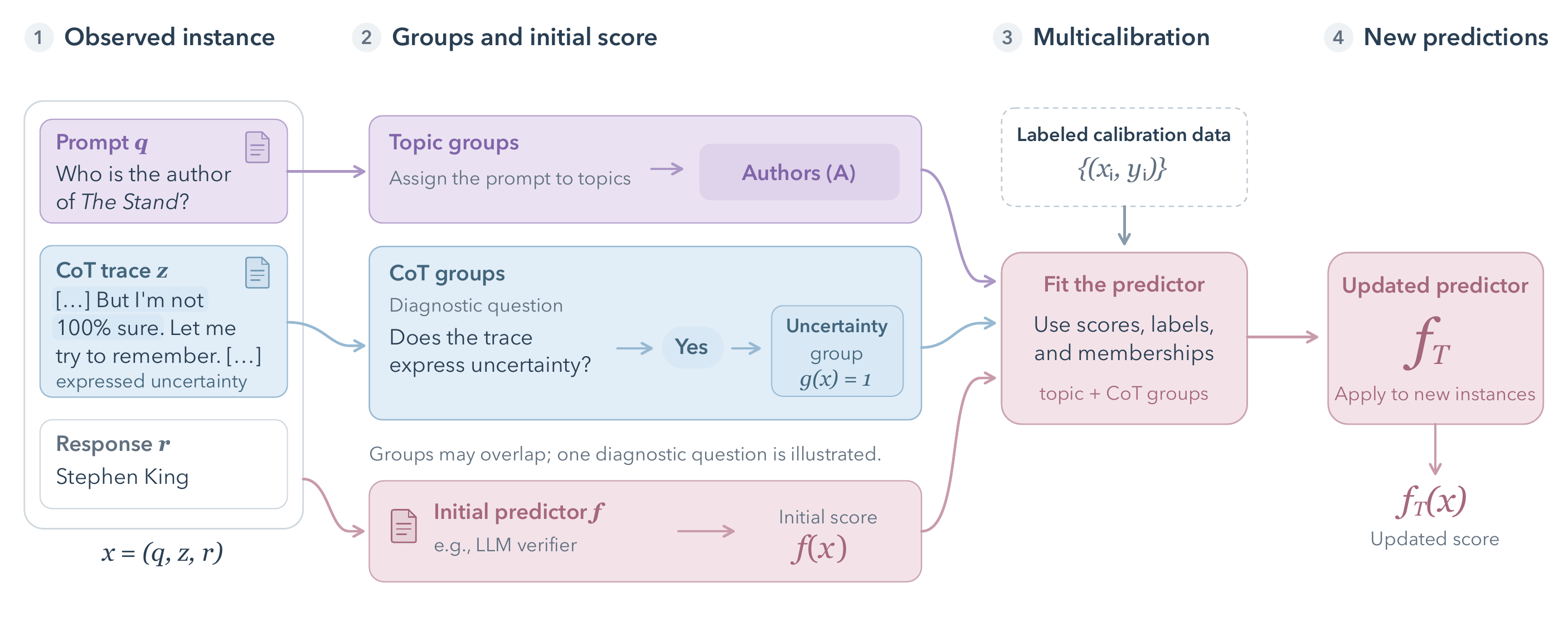}
        \caption{\textbf{Overview of reasoning-aware multicalibration.} Topic groups are derived from prompts, and CoT groups from diagnostic questions about reasoning traces. Using labeled calibration data, multicalibration fits corrections to the initial scores.}
    \label{fig:framework-overview}
\end{figure}

Our framework is a \emph{post-hoc} multicalibration protocol to refine a reliability predictor to perform meaningfully on any externally important sets of prompts/responses and leverages reliability information hidden in model reasoning. Despite using CoT, the LLM is not incentivized to manipulate its own reasoning, as the predictor is built on top of the model.
In summary (\Cref{fig:framework-overview}), a practitioner specifies a binary property of interest, e.g., hallucination (y/n), and an initial scoring function. Application-relevant groups, such as topics, specify where calibration guarantees of meaningfulness are necessary, while groups derived from reasoning traces can reveal remaining systematic prediction errors and improve the predictor further. Using labeled examples of the chosen property, a multicalibration process fits score corrections to the initial $f$ across both group families. %
The framework is intentionally agnostic to the specific instantiation of each step. 
In \Cref{sec:casestudy} we use our suggested instantiations to create a reasoning-aware multicalibrated predictor of hallucination.

\subsection{Initial scoring and Topic Groups}
\xhdr{Initial scoring.} The initial scoring function estimates the probability that the chosen binary property holds. This can be supplied by any task-specific detector (e.g., an LLM judge) or some function of token probability mass and does not need to be perfect, as the purpose of the framework is to refine it. However, ideally it should contain signal about the property of interest and vary across inputs. \Cref{app:scoring} describes how we implement the initial scoring function in our case study. %

\xhdr{Topic Groups.} Topic groups represent sets of prompts and responses for which it is important that the predictor be meaningful for external reasons; e.g., a practitioner may need a hallucination predictor to be multicalibrated with respect to questions from high-risk topics like medicine~\cite{multical_conf}. Each group defines a binary membership rule. These groups need not be disjoint (e.g., a prompt concerning a medical regulation may belong to both medicine and law). Topic groups can use existing domain categories or other application-relevant information. When labels are unavailable, we suggest constructing groups using embeddings of the question and other available context. \Cref{app:topic} describes how we construct and validate the topic groups used in our case study. %

\begin{figure}[t]
    \centering
    \includegraphics[width=\linewidth]{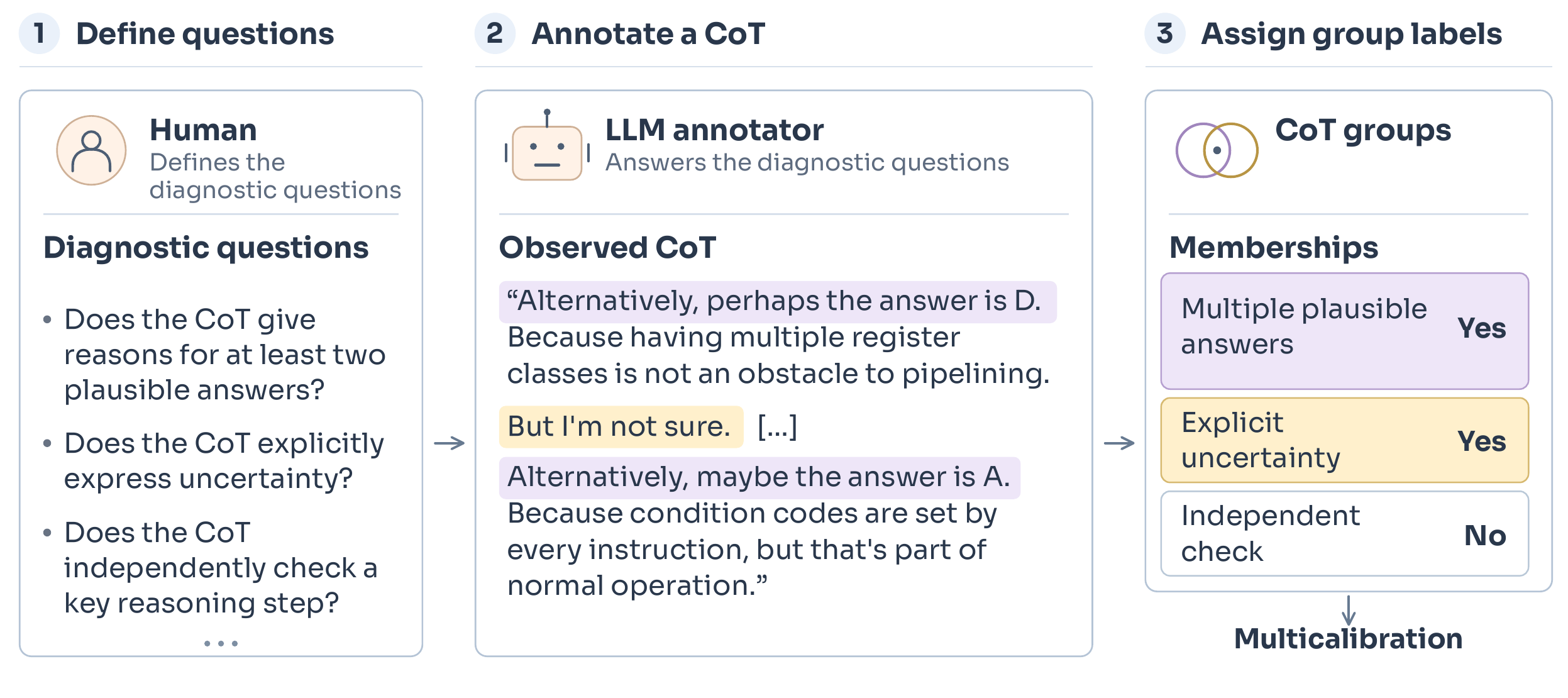}
    \caption{\textbf{Constructing CoT groups with diagnostic questions.} A practitioner reviews CoTs and defines the diagnostic questions. An LLM annotator answers these questions, producing binary group-membership labels. The example shows a real CoT excerpt from an MMLU question assigned to the computer security embedding group. See \Cref{tab:cot_questions} for the full list of diagnostic questions.}
    \label{fig:diagnostic_questions}
\end{figure}

\subsection{CoT Groups}\label{sec:cot-groups}

Within a topic, a predictor may be overconfident on responses whose CoT expresses uncertainty and underconfident on the rest, so the errors cancel in the topic average (Figure~\ref{fig:cot-hidden-errors}A). CoT groups can expose this.
We suggest constructing them through practitioner-selected \emph{diagnostic questions} about the %
CoT relevant to the reliability property, %
e.g., whether CoT expresses uncertainty, revises a claim, or independently checks a reasoning step. We present this method in \Cref{fig:diagnostic_questions} and direct the reader to \Cref{app:cot-groupings} for details.
%
%
%
Because practitioners select their questions, they can avoid assuming that the reasoning faithfully represents the model's internal computation and still derive useful groups even if the model's CoT becomes more strategic (as long as signal remains). For example, if, rather than producing long reasoning when unsure, the model (imperfectly) attempts to hide uncertainty by acting out-of-distributionally confident, diagnostic questions detecting this can still be valuable. 

\subsection{Multicalibrating over Topic and CoT}
The final step of our framework is to run the chosen initial scoring function, $f$, through a multicalibration algorithm w.r.t. the topic and CoT groups using a dataset of prompts, reasoning traces, responses, and reliability labels. Assuming the practitioner has such a domain-representative dataset, the algorithm can be any out-of-the-box batch multicalibration algorithm. As we have used it for our theoretical statements, %
we suggest iterative greedy histogram binning (IGHB), but alternative algorithms can also be used.
Of course, the need for a domain-representative dataset poses two potential limitations (consistent with all batch multicalibration applications): (1) where to get labeled data? (2) will the predictor generalize to other domains? We address these in \Cref{sec:discussion}. 
\section{Case Study: CoT for Hallucination Detection}
\label{sec:casestudy}
In this section, we apply reasoning-aware multicalibration to a hallucination predictor using Qwen3-8B~\cite{yang2025qwen3} across four different benchmarks: MMLU~\cite{hendrycks2021measuringmassivemultitasklanguage}, MMLU-Pro~\cite{wang2024mmlu}, TriviaQA~\cite{joshi2017triviaqa}, and PopQA~\cite{mallen2023trustlanguagemodelsinvestigating}. 
We instantiate the framework in \Cref{sec:framework}, with implementation details in Appendix~\ref{app:method}. We also examine alternative diagnostic labelers and multicalibration algorithms in Appendix~\ref{app:ablations}, and evaluate transfer across benchmarks in Appendix~\ref{app:cross-benchmark-transfer}.

Throughout, we compare predictors obtained by running IGHB on the same initial scores with different group sets: \emph{topic-only} ($\mathcal{G}=$ topic groups), \emph{CoT-only} ($\mathcal{G}=$ CoT groups), and \emph{joint} ($\mathcal{G}=$ topic $\cup$ CoT groups). As baselines, we report the \emph{uncalibrated} initial scores and \emph{marginal HB} (histogram binning with a single group containing all responses) which replaces each score by the empirical accuracy of its bin and is the strongest baseline that uses no group information. Independently of how a predictor is fit, we audit it on the topic groups, reporting the mean and worst-topic squared calibration error (gASCE; \Cref{def:sq_mcal}) alongside held-out Brier error (\Cref{def:brier}). Throughout, \emph{confidence} refers to the predicted probability that a prompt/response is \emph{not} a hallucination.

\subsection{CoT is Internally Important}

\begin{figure}[t]
    \centering
    \setlength{\abovecaptionskip}{6pt}
    \setlength{\parskip}{0pt}
    \includegraphics[width=0.85\linewidth]{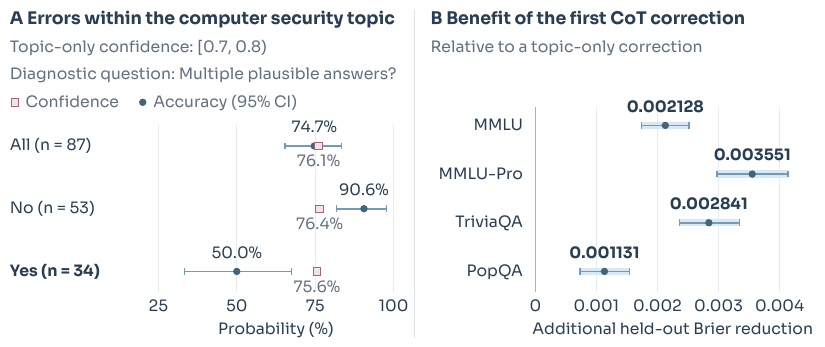}
    \caption{\textbf{CoT reveals opposing errors within a topic and useful score corrections.} (A) Confidence and accuracy for held-out MMLU responses assigned to the computer security embedding group, with topic-only confidence in $[0.7,0.8)$. (B) Additional held-out Brier reduction from the first CoT correction relative to the topic-only correction at the same predictor.}
    \label{fig:cot-hidden-errors}
\end{figure}

\label{sec:cot-corrections}
Section~\ref{sec:theory} motivates adding groups that expose useful corrections to a predictor and are internally important. We confirm that CoT-derived groups can serve such purpose. %
To first illustrate the useful information in CoT, Figure~\ref{fig:diagnostic_questions} highlights a sample in which CoT reveals the model seriously considering two competing answers before choosing an incorrect answer. Despite this signal, confidence multicalibrated with respect to only topic is as high as 79.0\%. 

Figure~\ref{fig:cot-hidden-errors}A shows that this phenomenon goes beyond a single example: among the held-out MMLU samples in the computer-security topic group where confidence multicalibrated w.r.t only topic is between $70\%$ to $80\%$, while the score seems calibrated (pooled mean confidence of 76.1\% is close to the accuracy of 74.7\%), the responses whose traces consider multiple plausible answers have 50.0\% accuracy at 75.6\% mean confidence, while the remaining 53 have 90.6\% accuracy at 76.4\% mean confidence. It is clear therefore that the CoT groups reveal errors that can largely cancel within this topic and confidence range.

Finally, \Cref{fig:cot-hidden-errors}B formally confirms that CoT-based groups are internally important by \Cref{def:improved_internal_groupings}:
for each dataset, \textbf{there exists some algorithmic step at which a CoT correction was preferred over a topic-based one}. We present the local improvement of the CoT patch over topic by applying two alternative corrections separately to this shared predictor: the correction selected by multicalibration w.r.t. topic and CoT and the next correction selected by topic-only multicalibration and reporting the difference in held-out Brier error reduction. %

\subsection{Improving Hallucination Detection}
\label{sec:improving}
\begin{table}[t]
    \centering
    \caption{Held-out Brier scores, lower is better. Reduction is Joint's percentage improvement over Topic-only.}
    \label{tab:cot-prediction-performance}
    \begingroup
\small
\setlength{\tabcolsep}{3pt}
\begin{tabular}{@{}lcccc@{}}
\toprule
Condition & \multicolumn{1}{c}{MMLU} & \multicolumn{1}{c}{MMLU-Pro} & \multicolumn{1}{c}{TriviaQA} & \multicolumn{1}{c@{}}{PopQA} \\
\midrule
Uncalibrated & 0.233 {\scriptsize $\pm$ 0.006} & 0.359 {\scriptsize $\pm$ 0.007} & 0.294 {\scriptsize $\pm$ 0.008} & 0.382 {\scriptsize $\pm$ 0.008} \\
Marginal HB & 0.137 {\scriptsize $\pm$ 0.004} & 0.184 {\scriptsize $\pm$ 0.004} & 0.209 {\scriptsize $\pm$ 0.004} & 0.148 {\scriptsize $\pm$ 0.004} \\
Topic-only & 0.127 {\scriptsize $\pm$ 0.003} & 0.163 {\scriptsize $\pm$ 0.004} & 0.185 {\scriptsize $\pm$ 0.004} & 0.114 {\scriptsize $\pm$ 0.004} \\
CoT-only & 0.113 {\scriptsize $\pm$ 0.003} & 0.141 {\scriptsize $\pm$ 0.004} & 0.164 {\scriptsize $\pm$ 0.004} & 0.125 {\scriptsize $\pm$ 0.004} \\
Joint & \textbf{0.112} {\scriptsize $\pm$ 0.003} & \textbf{0.139} {\scriptsize $\pm$ 0.004} & \textbf{0.163} {\scriptsize $\pm$ 0.004} & \textbf{0.109} {\scriptsize $\pm$ 0.004} \\
\midrule
Reduction & 11.3\% & 15.0\% & 12.4\% & 4.5\% \\
\bottomrule
\end{tabular}
\endgroup

\end{table}

Although we have confirmed that CoT-based groups are internally important, significant improvement is not guaranteed. We compare the held-out Brier error of a predictor multicalibrated with respect to CoT and topic (according to our framework) with that of topic-only and marginally calibrated predictors. We find that including CoT yields large improvements (Table~\ref{tab:cot-prediction-performance}).

Joint calibration, i.e., multicalibration w.r.t CoT and topics, achieves the lowest Brier score on all the four datasets. %
In particular, {\bf adding CoT groups reduces Brier score by 4.5--15.0\% relative to topic-only calibration}, while the reductions relative to marginal HB and the uncalibrated scores are even larger.
As we would expect, given topics are still included in the framework, improvements in accuracy do not come with a systematic worsening in topic-level calibration: joint calibration has similar mean and worst-topic calibration error to topic-only calibration, with neither method being uniformly better across benchmarks and both methods still achieving low topic calibration error (Table~\ref{tab:cot-topic-calibration}). Together, these results support using CoT-derived groups alongside topic groups: not only do they give us desirable multicalibration guarantees over CoT properties of interest, but they also contribute to improving probabilistic detection of hallucinations.

\subsection{Topic Calibration Recovery}
\begin{figure}[t]
      \centering
        \includegraphics[width=\textwidth]{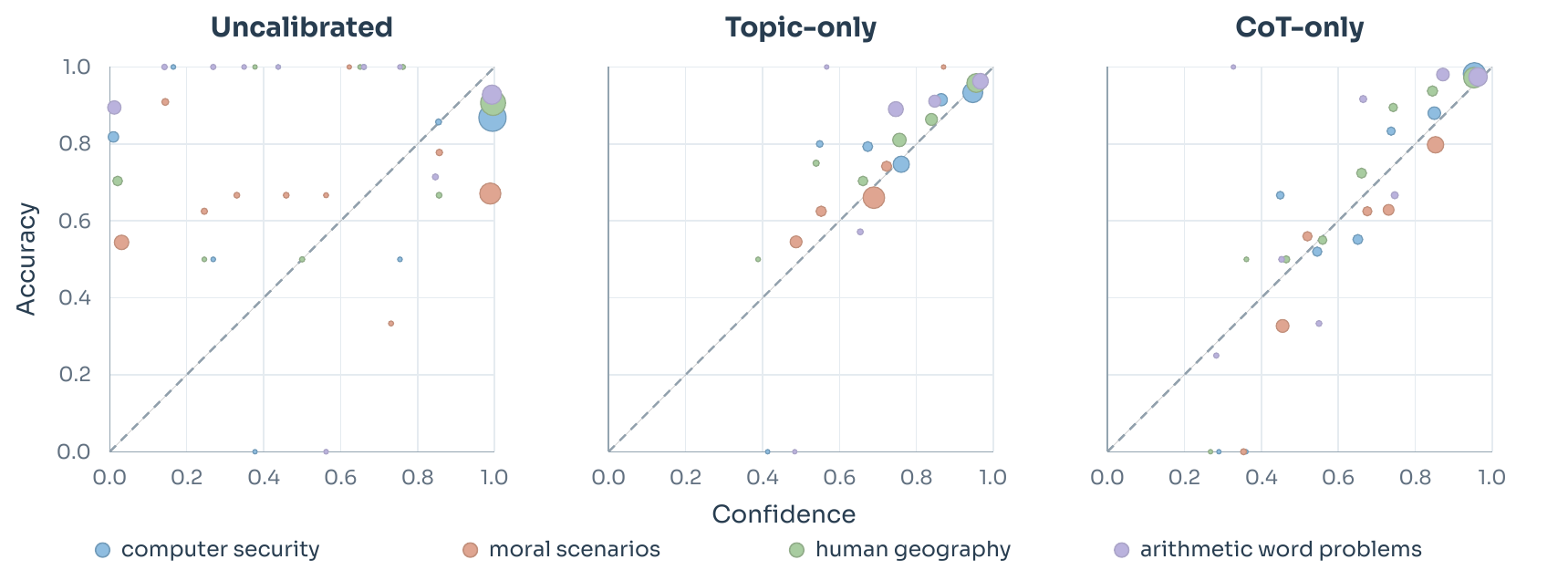}
    \caption{\textbf{CoT-only calibration improves calibration over topics.} The plot shows four MMLU topics. Each point gives mean confidence and accuracy within a topic--confidence bin; size indicates sample count. The diagonal marks agreement between confidence and accuracy. }
      \label{fig:topic-transfer-reliability}
  \end{figure}

Section~\ref{sec:cot-corrections} shows that CoT properties distinguish responses with opposing calibration errors within the same topic. If these properties identify similar errors across topics, correcting them could improve topic-level calibration implicitly, as measured by gASCE, \emph{without directly using topic memberships}. This would be interesting because it would suggest that some reliability properties defined over semantic groups can be addressed through checks on CoT alone. We explore this possibility by calibrating over CoT groups alone and evaluating the resulting predictions within topic groups.

\begin{wrapfigure}[]{R}{0.5\textwidth}
    \vspace{-\intextsep}
    \centering
    \setlength{\abovecaptionskip}{2pt}
    \setlength{\parskip}{0pt}
    \includegraphics[width=\linewidth]{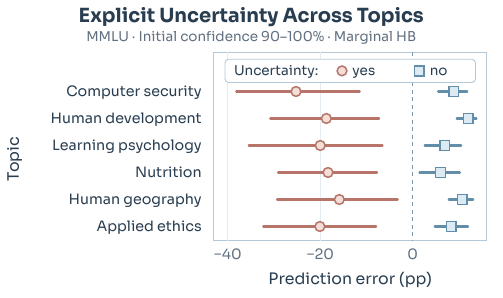}\par
    \nointerlineskip
    \caption{\textbf{CoT cues can reveal errors that marginal calibration misses}. The plot shows six MMLU topics at confidence of 90--100\%. Error is accuracy minus confidence.}
    \label{fig:topic-transfer-uncertainty}
\end{wrapfigure}

On all four datasets \textbf{CoT-only calibration improves mean and worst-topic calibration error over marginal calibration}, and is often comparable to topic only multicalibration (Table~\ref{tab:cot-topic-calibration}). %
We examine the reasons for such transfer in Figure~\ref{fig:topic-transfer-uncertainty}: in all six displayed topics, at initial confidence 90--100\%, HB overpredicts correctness for CoT traces that express uncertainty and underpredicts it for traces that do not (for clarity, we show one diagnostic question and six MMLU topics; Figure~\ref{fig:topic-transfer-cues-all} in the appendix gives all comparisons across datasets, topics, and diagnostic questions). A correction based on this diagnostic label could move predictions in the right direction across several topics without using their memberships. This shared error structure suggests that the predictor's errors depend not only on what the model is asked about, but also on how it reasons over them.

Figure~\ref{fig:topic-transfer-reliability} shows this behavior in the form of a calibration plot. Its three panels show topic-level reliability before calibration, after topic-only calibration, and after CoT-only calibration. Relative to the base predictor, CoT-only calibration moves many of the displayed topic-confidence bins closer to the diagonal (representing perfect multicalibration), despite the corrections being fitted without using semantic information. The topic-only panel provides a direct comparison with corrections fitted using topic memberships. For readability, we display the largest topic in each of the 4 largest families of topics with similar question vocabulary; Figure~\ref{fig:topic-transfer-all-topics} shows all the topic groups.

\section{Discussion and Limitations} \label{sec:discussion}
We investigate whether chain-of-thought (CoT) traces can be used to construct meaningful predictors of LLM response reliability, including hallucination, and safety, while preserving multicalibration across application-relevant groups such as topic.
Motivated by the fact that CoT has signal for these properties, we present and empirically validate a framework of reasoning-aware multicalibration. %
The framework is particularly appealing as it does not introduce any training incentives to the LLM and may harness signals of various kinds. The most significant limitation is one that is shared across all LLM multicalibration regimes: the necessity of \emph{domain-representative calibration data}. In particular, the use of such data poses two questions, (1) how could a practitioner get such a dataset? and (2) given the dataset should be representative of some particular domain, how well should the practitioner expect his calibrated predictor to generalize? We conclude with an elaboration upon these questions.

\xhdr{Where to get calibration data?} We suggest (and use in our case study) batch multicalibration on a domain-representative dataset of prompts, responses, CoTs, and reliability property labels. For various domains or extremely broad ones, getting such a dataset may be difficult. In cases where only labels are missing, a practitioner may consider using a more sophisticated LLM-as-judge to assign reliability labels, though this should be done with caution as an LLM may not be perfectly aligned with human preferences on reliability. In cases where no batched dataset exists, but prompt/response/CoT are generated over time (perhaps as users interact with the LLM) practitioners can consider using online multicalibration methods.

\xhdr{Generalization.} The process of multicalibration is relatively simple, it refines a predictor such that it is more representative of the conditional distribution of the binary property over some distribution. Unfortunately, this can cause performance issues if, for example, a hallucination predictor created for an LLM's responses to trivia questions is used for its discussion of politics. To some extent this is unavoidable, but our framework's use of CoT groups may somewhat decrease this issue when comparing to a multicalibrated predictor with respect to only topic. Notice that while rates of various reliability properties for the same topic across different domains likely change, e.g., safety of chemistry prompts in the domain of trivia v.s. safety of chemistry prompts in the domain of home projects, signals from CoT could be more consistent and thus relatively generalizable across domains. Continuing with the safety example, CoT language cues that suggest the potential for harm, i.e., the LLM theorizing in its CoT that the output may be misused in a harmful way, would represent groups of prompt/responses that are likely unsafe \emph{across domains}. Notably we can already see that our CoT groupings generalize across domains relatively well because across all datasets in \Cref{sec:casestudy}, we use the same CoT diagnostic questions and find they are useful across the board.
To further check the generalization of reasoning-aware predictors, in \Cref{app:cross-benchmark-transfer} we look at predictors' performance on benchmark datasets they were not multicalibrated on. We find that while, of course, each predictor performs best on its own benchmark, predictors that are multicalibrated w.r.t only CoT and even those that are multicalibrated w.r.t. both CoT and topic transfer well (\Cref{tab:cross-benchmark-brier}).

%

%

%

\subsection*{AI use statement}
In this work, we used generative AI tools for retrieval/discovery, specifically Claude and ChatGPT were used in the search for works related to Multicalibration theory, LLM calibration, and CoT. All LLM suggested related work were looked over by an author and the description of each work and its relation to ours was human written with minimal AI assistance, i.e., AI was used for edits and as a question/answer tool about parts of papers that were unclear. Generative AI tools were also used in research execution: specifically, Codex was used to assist with writing, debugging, and running experimental code. All AI-assisted code and resulting output was manually inspected. Lastly, generative AI tools were used to polish the writing of author-written text in some parts of the paper and to improve figures visually (i.e., the basics of figures were human generated and an AI was used to produce a nicer and cleaner version). Any text written by AI in this final product was a result of having it looking over the draft and making edits that we approved of. We have not used generative AI tools to draft the initial sections of the paper or to prove theoretical claims/results. 

\bibliography{bib}
\bibliographystyle{abbrvnat}

\clearpage
\appendix
\section{Additional Related Work}\label{app:extended_related_works}

\xhdr{Extended discussion of optimal groupings.}
Two related works have analyzed groups in multicalibration with a similar focus on good squared error. \citet{burhanpurkar2021scaffoldingsets} construct a small collection of sets such that multicalibration with respect to them yields a predictor close to the Bayes optimum. \citet{mcal_boosting} give a complementary characterization in terms of weak learning: multicalibration with respect to a hypothesis class implies Bayes optimality exactly when said hypothesis class satisfies an appropriate weak-learning condition. Though conceptually similar, our framing differs in two key ways: we focus on (1) (marginal) group augmentation, i.e., having an initial $\mathcal{G}$ a priori and adding groups and (2) application-motivated groups that will not necessarily satisfy weak learning guarantees.

\xhdr{Details on LLM calibration works.}
A large amount of work has been dedicated to the problem of ensuring confidence scores outputted by LLMs are well-calibrated (we refer to \citet{llm-calibration-survey} for a survey). 
Evaluating question/answer confidence using a confidence metric that is a function of token probabilities, \citet{kadavath2022languagemodelsmostlyknow} find that large models they test (800M to 52B parameters) are generally well-calibrated in their token probabilities on a suite of multiple-choice benchmarks. However this kind of calibration is not true of all models, testing T5, BART, and GPT-2, \citet{know_what_language_models_know} find these models' probabilities are \emph{not} well-calibrated on question-answering (QA) tasks, though temperature-based scaling and feature-based decision tree calibration methods are helpful. Using the intuition that models achieve more calibrated confidences when they have distinct choices from \cite{kadavath2022languagemodelsmostlyknow}, \citet{kapoor-etal-2024-calibration} propose a method to directly calibrate logits, which does not require a set of candidate answers like the methods used by \citet{know_what_language_models_know}. Other works study LLM calibration focusing on confidence metrics that are not token probability-based: \citet{lin2022tmlr-teaching} calibrate (via finetuning) GPT-3's \emph{verbal uncertainty} on a novel mathematics benchmark, \cite{zhang-etal-2024-calibrating} propose a new \emph{sampling-based} confidence measure that does not require access to tokens, and similarly \cite{lyu_calibrating_2025} also evaluate three \emph{sampling consistency}-based confidence metrics. Yet another relevant strand of works examine how calibration is affected by various aspects of LLM creation and use: \citet{zhu-etal-2023-calibration} examine the changes in LLM calibration \emph{during the training cycle}, \citet{just_ask_calibration} analyze model calibration, specifically verbalized probabilities \emph{after RLHF} and \cite{damani2026binaryrewardstraininglms} train for that, and \citet{zhao-etal-2024-fact} evaluate calibration (across a variety of confidence metrics) when \emph{different prompting techniques} are employed.

\xhdr{Chain of thought.}
Chain of thought reasoning is the generation of intermediate steps an LLM takes before producing a final output, popularized initially as a prompting technique by \citet{cot_prompting}. Today, many models engage in, and present alongside the final response, a similar reasoning~\cite{openai_system_card,anthropic_system_card,deepseek_reasoning}. 
CoT has been observed to contain signal about general reliability properties of the output: \citet{wang2026joint} detect hallucinations using answer uncertainty and consistency across sampled reasoning traces; \citet{arnav2026cot} find that reading CoT helps a monitor detect harmful changes to code (but can also mislead it when the trace presents those changes as harmless); \citet{baker2025monitoring} show that access to CoT improves reward-hacking detection compared with observing actions and outputs alone. We leverage this signal in the creation of our multicalibration framework.  At the same time, many works have shown that reasoning traces can be incomplete or misleading. 
\citet{dont_say_what_they_think} show that biasing cues can influence models' answers without being acknowledged in their explanations; \citet{chen2025reasoning} find similar omissions in models trained for extended reasoning. Other studies examine whether models use their stated reasoning: \citet{paul-etal-2024-making} find that changing intermediate reasoning does not consistently change the answer, while \citet{arcuschin2025chain} observe contradictory justifications and correct answers supported by invalid reasoning without injected hints/traces. There are also limits to what can be inferred from the reasoning text: \citet{du2026legibility} find that LLM judges only partly infer a reasoning step's effect on correctness from its text, while \citet{baherwani2026llmreasoningvisiblechainofthought} show that even semantically irrelevant filler tokens can improve performance on synthetic reasoning tasks. Training can further encourage models to conceal information \citep{baker2025monitoring, ICLR2026_53813b70}.
\section{Additional Theoretical Details}\label{app:theory}

\subsection{Additional Details}
Unfortunately, adding internally important groups does not guarantee the desired \emph{global} Brier improvement. We provide an example next. 

\begin{example}[Predicted value collapse]\label{ex:collapse}
    Using the setting from \Cref{ex:tree}, consider an alternative group set, $\mathcal{G}':= \{g_1, g_2\}$ where $g_2(x) = 1$ $\forall x \in \mathcal{X}$. Clearly, $\mathcal{G}\subset \mathcal{G}'$. Using the IGHB version of \cref{alg:iterative_greedy} with $\alpha = .2$, \Cref{fig:bad_path} presents all the possible patches. The addition of $g_2$ expands upon the patching options and provides a branch in which $f_T$ may become the \emph{optimal} predictor (i.e., $f_T= (1,0,.6)$). However, because \cref{alg:iterative_greedy} takes the greedy step at each round, the middle path will be taken and this $f_T$ will be \emph{worse} than that of \Cref{ex:tree}.

\end{example}

\begin{figure}[htbp]
    \centering
    \includegraphics[width=0.8\linewidth]{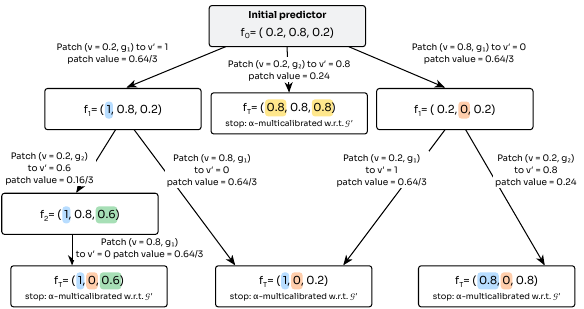}
    \caption{All potential patches (using IGHB patch value) for \Cref{ex:collapse}.}
    \label{fig:bad_path}
\end{figure}

\subsection{Proofs}
\begin{proof}[Proof of \Cref{prop:nec_cond}]
    Let $\mathcal{P}$ be the patching path from $\alpha$-multicalibration using an iterative greedy multicalibration algorithm on $f$ with respect to $\mathcal{G}$ and $\mathcal{P}'$ be the patching path w.r.t $\mathcal{G}\cup \mathcal{H}$. Assume $\mathcal{H}$ satisfies \Cref{def:added_groupings_goal}. By definition of an iterative greedy algorithm, every patch in $\mathcal{P}$ is associated with a strict decrease in Brier error. Therefore, because $B(f_{T})> B(f_{T'})$, it must be the case that $\mathcal{P}\neq \mathcal{P}'$. At least one of the following must be true: (A) $\exists j \in \left[|\mathcal{P}|\right]$ s.t. $\mathcal{P}'_{j}\neq \mathcal{P}_j$, (B) $\forall j \in \left[|\mathcal{P}|\right]$, $\mathcal{P}'_{j}= \mathcal{P}_j$ and $|\mathcal{P}| < |\mathcal{P}'|$ where $\mathcal{P}_{j}$ and $\mathcal{P}'_{j}$ refer to the $j$th patches of $\mathcal{P}$ and $\mathcal{P}'$ respectively and $\mathcal{P}'_{j} := \{\}$ if $j> |\mathcal{P}|$. We show that if at least (A) or (B) hold, then at least one of the conditions of \Cref{def:improved_internal_groupings} holds.

    Suppose (A) holds. Let $j^\star=\min\{j \in [|\mathcal{P}|]: \mathcal{P}_j'\neq \mathcal{P}_j\}$. By definition of $j^\star, \mathcal{P}$, and $\mathcal{P}'$, before patch $j^\star$, all $\mathcal{G}$ patches available in the multicalibration with respect to $\mathcal{G}$, must also be available in the multicalibration with respect to $\mathcal{G}\cup\mathcal{H}$. However, for patch $j^\star$, the algorithm w.r.t $\mathcal{G}\cup\mathcal{H}$ proceeds with a different patch. Therefore, by definition of \Cref{alg:iterative_greedy}, it must be the case that there was a $\mathcal{H}$ patch whose patch value was (strictly) larger than any available $\mathcal{G}$ patch. Equivalently, because the patch value of the $\mathcal{P}_{j^\star}$ must be the largest available $\mathcal{G}$ patch, this gives condition 1 of \Cref{def:improved_internal_groupings}.

    Suppose (B) holds. Let $j^\star = |\mathcal{P}|$. Because $j^\star$ represents the last patch for a greedy iterative algorithm w.r.t. $\mathcal{G}$ and $\mathcal{P}_j = \mathcal{P}'_j$ for all indices up to and including $j^\star$, by definition of \Cref{alg:iterative_greedy}, after patch $j^\star$ (w.r.t $\mathcal{G}$ OR $\mathcal{G}\cup\mathcal{H}$), the function must be $\alpha$-multicalibrated w.r.t $\mathcal{G}$. However, $|\mathcal{P}| < |\mathcal{P}'|$, meaning another step when multicalibrated w.r.t. $\mathcal{G}\cup\mathcal{H}$ was taken. Therefore, it must be the case that condition 2 of \cref{def:improved_internal_groupings} holds.
\end{proof}
\section{Additional Methodological Details}\label{app:method}

\subsection{Datasets and Response Generation}

We use four question-answering benchmarks: MMLU and MMLU-Pro for multiple-choice questions, and TriviaQA and PopQA for short-answer questions. We use the full test splits of MMLU, MMLU-Pro, and PopQA. For TriviaQA, we use the validation split, retaining one example per unique question identifier. Questions are presented without supporting documents or retrieval.

We generate one response per question using Qwen3-8B with thinking enabled. For MMLU, we use four or five subject-specific reasoning few-shot examples from Chain-of-Thought Hub\footnote{\url{https://github.com/FranxYao/chain-of-thought-hub}}. For MMLU-Pro, we use five category-specific demonstrations from the dataset's official validation split. Both prompts ask the model to reason step by step and end with ``The answer is (X),'' where X identifies the selected option. TriviaQA and PopQA use zero-shot prompts that request step-by-step reasoning. Answers are generated with sampling temperature 0.

The binary target indicates whether the extracted answer is correct. For multiple-choice questions, we compare the extracted option with the reference option. For both short-answer datasets, the extracted answer must exactly match a reference answer or an accepted alias after normalization. Accuracy results for this setup are reported in Table~\ref{tab:answer-accuracy}. 

\begin{table}[htbp]
    \centering
    \caption{Qwen3-8B answer accuracy and number of held-out evaluation examples for each dataset.}
    \label{tab:answer-accuracy}
    \begin{tabular}{lrr}
        \toprule
        Dataset & Evaluation examples & Answer accuracy \\
        \midrule
        MMLU & 4,126 & 83.28\% \\
        MMLU-Pro & 3,438 & 75.36\% \\
        TriviaQA & 2,938 & 61.98\% \\
        PopQA & 4,107 & 21.09\% \\
        \bottomrule
    \end{tabular}
\end{table}

\subsection{Scoring Function}
\label{app:scoring}

For the initial scoring function, we use the True/False softmax scoring approach of \citep{multical_conf}. In a separate call, we provide Qwen3-8B with the original question and its extracted answer, including the answer choices for multiple-choice questions. The prompt asks: ``Is the proposed answer correct? Answer with exactly one word, True or False.''
The initial reliability score is then computed as 
\[
    f_0(x)=\frac{p_{\mathrm{True}}(x)}{p_{\mathrm{True}}(x)+p_{\mathrm{False}}(x)},
\]
where $x$ denotes the example and $p_{\mathrm{True}}(x)$ and $p_{\mathrm{False}}(x)$ are the next-token probability masses assigned to True and False under the scoring prompt. We compute these masses using the 20 highest-probability tokens returned by the model. Tokens outside this set contribute zero; if neither answer is represented, the score is treated as missing. 

\subsection{Topic Groups}
\label{app:topic}

Throughout the paper, ``topic groups'' refers to groups constructed from semantic embeddings of questions and their answer choices. This provides a general procedure that applies across datasets without requiring predefined subject labels. 

We construct the groups separately for each dataset. We encode each question and its answer choices, when present, with UAE-Large-V1 \cite{li-li-2024-aoe}, using a maximum input length of 512 tokens, and normalize the embeddings to unit length. We reduce the embeddings to 20 dimensions using principal component analysis (PCA) and fit a Gaussian mixture model with diagonal covariance matrices. We select the number of components between 2 and 60 by minimizing the Bayesian information criterion on the group-construction sample. Each component defines a topic group: an example belongs to it when its posterior probability for that component is at least the mean posterior on the group-construction sample. An example can therefore belong to multiple groups. 

We fit the dimensionality reduction, mixture, and membership thresholds on a separate sample using only question content, then hold them fixed for calibration and evaluation. This yields 39 groups for MMLU, 42 for MMLU-Pro, 27 for TriviaQA, and 37 for PopQA. For topic-level evaluation, we retain groups containing at least 1\% of the calibration sample: all groups for the first three datasets and 34 for PopQA.

\begin{figure*}[htbp]
    \centering
    \includegraphics[width=\textwidth]{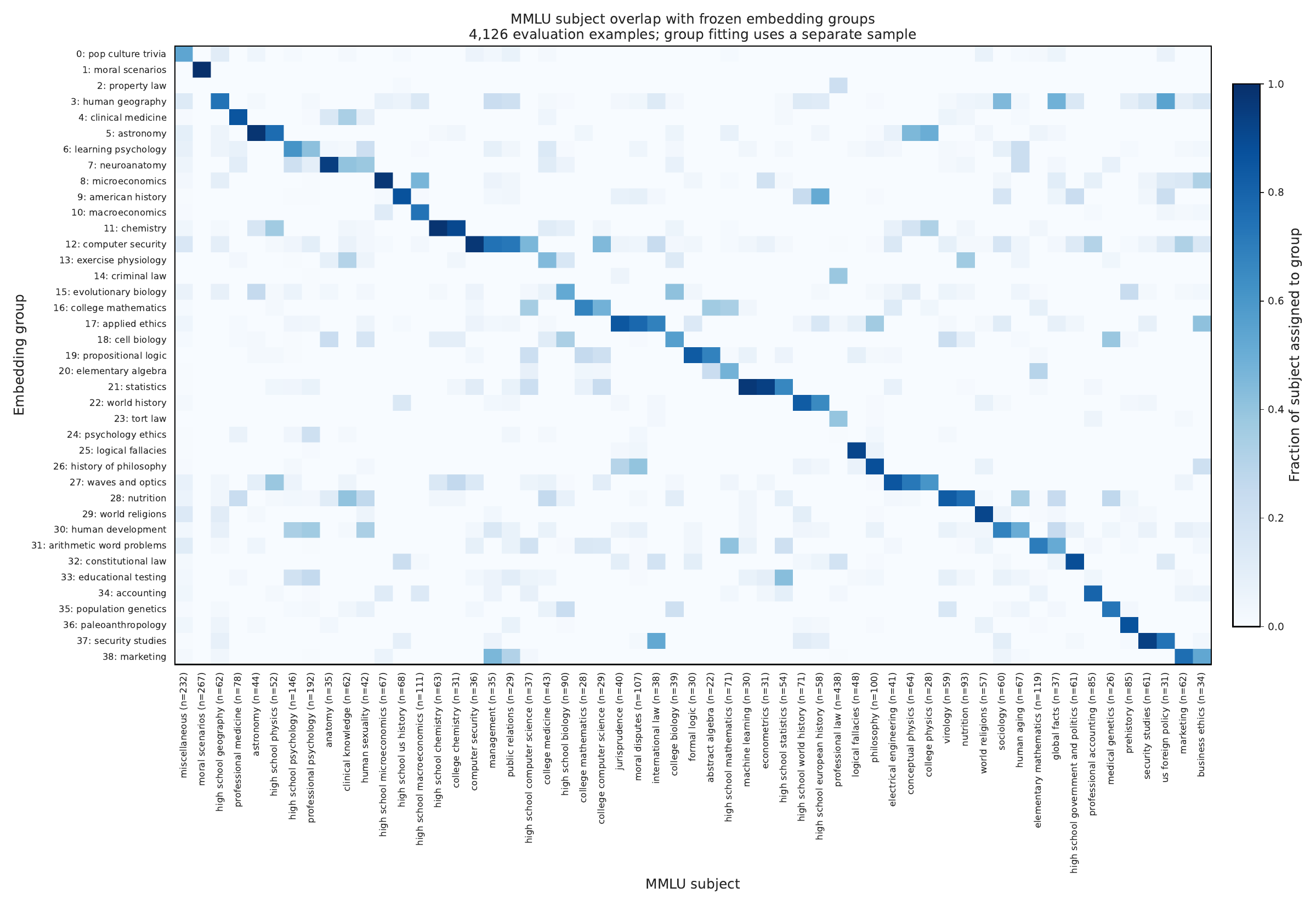}
    \caption{Overlap between embedding-based groups and MMLU subjects. Each cell shows the fraction of examples in a subject assigned to an embedding group. Subjects are ordered by their largest group overlap. Memberships can overlap, so a column need not sum to one. Embedding clusters are named using posterior-weighted c-TF-IDF
keywords and manual inspection of representative questions.
Names do not affect group membership or calibration.}
    \label{fig:topic-subject-overlap}
\end{figure*}

On MMLU, which provides subject labels, we examine how the resulting groups align with the given categories for validation in an evaluation set. We find that the groups generally capture subject structure while combining related categories. For example, the chemistry group contains 62 of 63 high-school chemistry questions and 28 of 31 college chemistry questions. Memberships overlap: 37.7\% of examples belong to multiple groups, averaging 1.52 memberships per example. The highest-probability cluster assignments have adjusted mutual information of 0.614 with the subject labels. Across binary groups, mean mutual information with subjects is 0.1267 bits, compared with 0.0109 under 999 subject-label permutations ($p=0.001$). We name embedding clusters using class-based c-TF-IDF \citep{grootendorst2022bertopicneuraltopicmodeling}, which scores terms by their prominence within a cluster relative to other clusters. We weight each question's contribution by its posterior cluster probability and inspect representative questions manually. These names are used only for visualization and do not affect group membership or calibration (Figure~\ref{fig:topic-subject-overlap}).

\subsection{CoT Groups}
\label{app:cot-groupings}

Our main experiments construct CoT groups using the ten diagnostic questions in Table~\ref{tab:cot_questions}. These questions describe visible properties of the reasoning (e.g., expressed uncertainty, consideration of competing answers, self-correction, and verification). We use the same questions across all four datasets. Each question defines one binary group containing the traces for which the labeler answers ``Yes'' or ``No''. The resulting ten groups overlap, as a reasoning trace can exhibit several properties.

We use Gemma4-31B as the diagnostic labeler. For each example, we pass the model the reasoning trace alone and the 10 diagnostic questions; the original question, separately extracted final answer, reference answer, correctness label, and initial reliability score are not part of the input. Labeling uses temperature 0, and group membership is the labeler's binary judgment. Appendix~\ref{app:ablations} repeats the diagnostic labeling with Qwen3-8B and Qwen3-4B for robustness.

We stress that the diagnostic labels describe the observed text, but do not establish that a trace faithfully reports the model's internal reasoning, or that a stated check or correction is valid. Their value depends on whether these observable properties distinguish systematic errors in the initial reliability score. 

\begin{table}[htbp]
  \centering
  \small
  \renewcommand{\arraystretch}{1.12}
  \begin{tabular}{@{}p{0.24\linewidth}p{\dimexpr0.66\linewidth-4\tabcolsep\relax}r@{}}
    \toprule
    \raggedright CoT property & Diagnostic question & Yes \\
    \midrule
    \raggedright Explicit uncertainty & Does the CoT explicitly express uncertainty about its final answer or about a fact needed to choose the answer? & 7,723 \\
    \addlinespace[4pt]
    \raggedright Explicit certainty & Does the CoT explicitly state high confidence or certainty about its final answer or a key premise? Do not count merely calling an option ``correct.'' & 420 \\
    \addlinespace[4pt]
    \raggedright Knowledge gap or guess & Does the CoT say that it does not know, only weakly remembers, assumes, or guesses a fact needed to answer? & 5,195 \\
    \addlinespace[4pt]
    \raggedright Multiple plausible answers & Does the CoT seriously treat at least two distinct answers or interpretations as plausible, giving some reason for each? Do not count routine enumeration of options. & 6,566 \\
    \addlinespace[4pt]
    \raggedright Reasoned elimination & Does the CoT eliminate at least one candidate answer using an explicit reason? & 7,479 \\
    \addlinespace[4pt]
    \raggedright Self-correction & Does the CoT retract or materially revise an earlier claim, calculation, inference, or tentative answer? Do not count discourse markers such as ``wait'' without a revision. & 3,940 \\
    \addlinespace[4pt]
    \raggedright Independent check & Does the CoT check a key step through recomputation, substitution, a second method, a counterexample, or comparison with stated evidence? Do not count merely restating the claim. & 2,414 \\
    \addlinespace[4pt]
    \raggedright Problem with the question & Does the CoT claim that the question, context, or answer choices are ambiguous, malformed, incomplete, or inconsistent? & 1,183 \\
    \addlinespace[4pt]
    \raggedright Uncertainty with a decisive answer & Does the CoT end with an unqualified answer while leaving its own stated uncertainty or competing candidates unresolved? & 6,857 \\
    \addlinespace[4pt]
    \raggedright Unresolved internal conflict & Does the CoT identify a contradiction or conflict among its own claims, calculations, or candidate conclusions and fail to resolve it before answering? & 402 \\
    \bottomrule
  \end{tabular}
  \caption{Diagnostic questions used to construct reasoning-aware groups. Each question is answered Yes or No based on the observed CoT. Yes counts are pooled over the 14,609 evaluation traces across MMLU, MMLU-Pro, TriviaQA, and PopQA; groups overlap. The ``Yes'' column reports the distribution of positive answers across the questions.}
  \label{tab:cot_questions}
\end{table}

\xhdr{Extensions.}
Diagnostic questions are one of many ways to construct CoT groups within our framework. We chose them because they describe reasoning properties that apply across question formats, and are easily generalizable to different setups. Other constructions can use any observable feature of the trace to define group membership.

As a simple alternative, we use four text features requiring no diagnostic-labeling calls: the length of the CoT trace, and the densities of hedging, contrast, and tentative-commitment expressions. Densities count case-insensitive phrase matches per word, such as ``perhaps,'' ``however,'' and ``I think,'' respectively. These features are extracted via regex. For each of them, we choose a threshold among midpoints of consecutive distinct values to minimize within-group squared error of correctness labels on the separate group-construction sample. Values at or below the threshold and those above it define two complementary groups, yielding eight groups per dataset, over which we calibrated. Results are shown in Table~\ref{tab:cot-grouping-comparison}.

\begin{table}[htbp]
    \centering
    \begingroup\small
\setlength{\tabcolsep}{4pt}
\renewcommand{\arraystretch}{1.12}
\begin{tabular}{@{}lrrrr@{}}
\toprule
Condition & MMLU & MMLU-Pro & TriviaQA & PopQA \\
\midrule
\multicolumn{5}{l}{\textbf{Brier score} $\downarrow$} \\
Uncalibrated & 0.233 {\scriptsize $\pm$ 0.006} & 0.359 {\scriptsize $\pm$ 0.007} & 0.294 {\scriptsize $\pm$ 0.008} & 0.382 {\scriptsize $\pm$ 0.008} \\
Marginal HB & 0.137 {\scriptsize $\pm$ 0.004} & 0.184 {\scriptsize $\pm$ 0.004} & 0.209 {\scriptsize $\pm$ 0.004} & 0.148 {\scriptsize $\pm$ 0.004} \\
Topic-only & 0.127 {\scriptsize $\pm$ 0.003} & 0.163 {\scriptsize $\pm$ 0.004} & 0.185 {\scriptsize $\pm$ 0.004} & 0.114 {\scriptsize $\pm$ 0.004} \\
CoT-only (diagnostics) & 0.113 {\scriptsize $\pm$ 0.003} & 0.141 {\scriptsize $\pm$ 0.004} & 0.164 {\scriptsize $\pm$ 0.004} & 0.125 {\scriptsize $\pm$ 0.004} \\
CoT-only (text features) & 0.118 {\scriptsize $\pm$ 0.003} & 0.157 {\scriptsize $\pm$ 0.003} & 0.175 {\scriptsize $\pm$ 0.004} & 0.132 {\scriptsize $\pm$ 0.004} \\
Joint (diagnostics) & \textbf{0.112} {\scriptsize $\pm$ 0.003} & \textbf{0.139} {\scriptsize $\pm$ 0.003} & \textbf{0.163} {\scriptsize $\pm$ 0.004} & \textbf{0.109} {\scriptsize $\pm$ 0.004} \\
Joint (text features) & 0.118 {\scriptsize $\pm$ 0.003} & 0.148 {\scriptsize $\pm$ 0.003} & 0.171 {\scriptsize $\pm$ 0.004} & 0.113 {\scriptsize $\pm$ 0.004} \\
\bottomrule
\end{tabular}\endgroup

    \caption{Brier scores for diagnostic and text-feature CoT groups using IGHB, with shared scores, examples, splits, and topic groups.}
    \label{tab:cot-grouping-comparison}
\end{table}

Further alternatives may include clustering embeddings of CoTs, or defining groups around task-specific reasoning steps. While these constructions can be incorporated into the same calibration procedure, we do not evaluate them here.

\section{Additional Experimental Details}\label{app:experiments}

Table~\ref{tab:cot-topic-calibration} reports mean and worst-topic calibration error for the predictors compared in \Cref{sec:casestudy}.

Figure~\ref{fig:cot-transfer-all-groups} complements the topic-level audit in Figure~\ref{fig:topic-transfer-all-topics} by evaluating the same predictors within CoT groups. Figure~\ref{fig:heatmap} reports errors within CoT subgroups across all 39 MMLU topic groups after topic-only calibration.


\begin{figure*}[htbp]
    \centering
    \includegraphics[width=\textwidth]{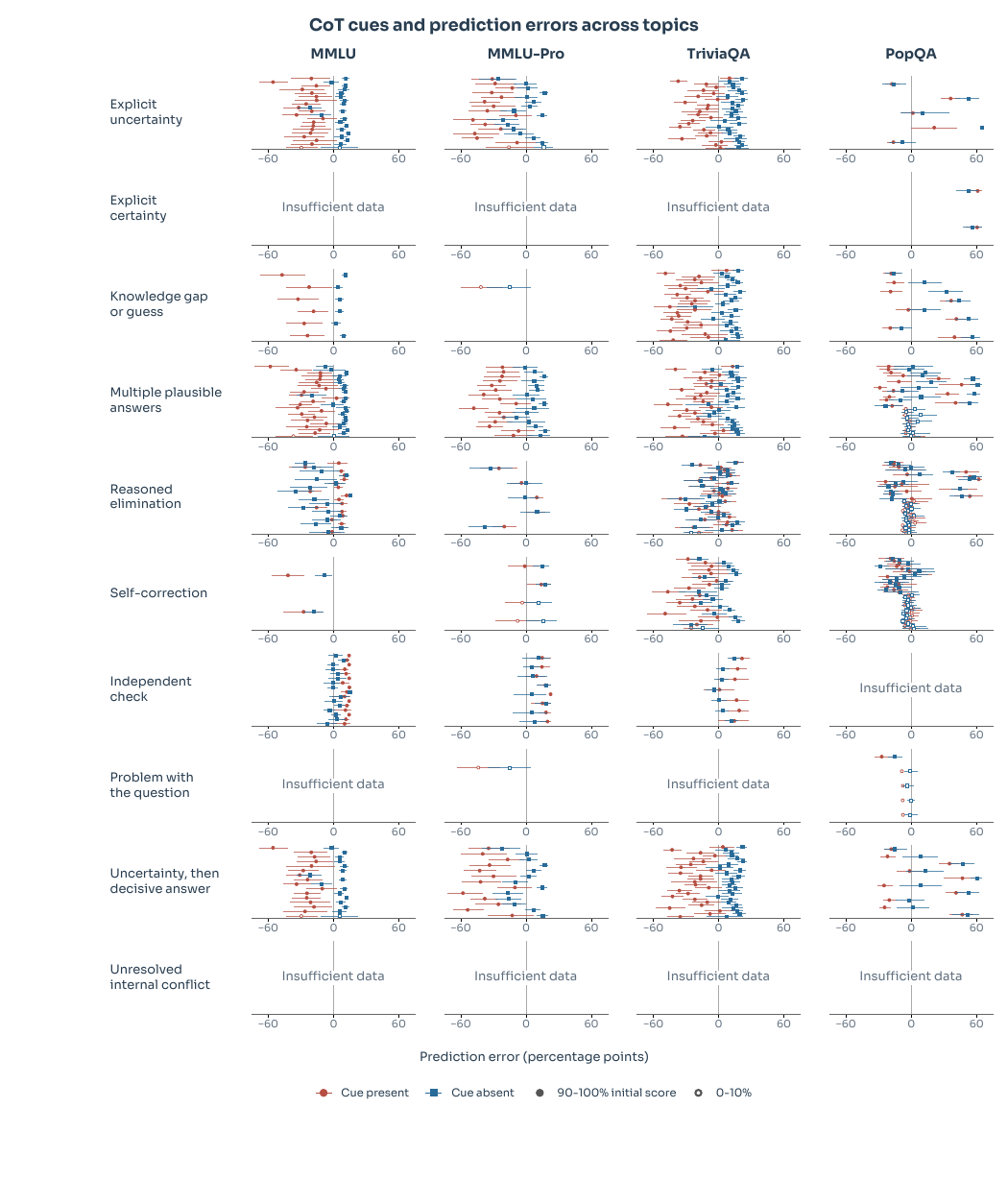}
    \caption{CoT-associated errors after marginal histogram binning. Each panel identifies a diagnostic, and each horizontal row within a panel is one topic and initial-score bin. Prediction error is accuracy minus mean predicted confidence, in percentage points. Each subgroup must have at least 20 fitting and 20 evaluation examples, otherwise we mark the panel as ``insufficient data''. }
    \label{fig:topic-transfer-cues-all}
\end{figure*}

\begin{table*}[htbp]
    \centering
\begingroup
\small
\setlength{\tabcolsep}{5pt}
\begin{tabular}{@{}lrrrr@{}}
\toprule
Condition & MMLU & MMLU-Pro & TriviaQA & PopQA \\
\midrule
\multicolumn{5}{l}{\textbf{Mean topic gASCE} $\downarrow$} \\
Uncalibrated & 0.106 {\scriptsize $\pm$ 0.005} & 0.216 {\scriptsize $\pm$ 0.008} & 0.115 {\scriptsize $\pm$ 0.007} & 0.274 {\scriptsize $\pm$ 0.009} \\
Marginal HB & 0.011 {\scriptsize $\pm$ 0.001} & 0.027 {\scriptsize $\pm$ 0.002} & 0.022 {\scriptsize $\pm$ 0.002} & 0.041 {\scriptsize $\pm$ 0.003} \\
Topic-only & \textbf{0.008} {\scriptsize $\pm$ 0.001} & 0.015 {\scriptsize $\pm$ 0.002} & 0.014 {\scriptsize $\pm$ 0.002} & \textbf{0.009} {\scriptsize $\pm$ 0.002} \\
CoT-only & 0.010 {\scriptsize $\pm$ 0.001} & 0.018 {\scriptsize $\pm$ 0.002} & 0.017 {\scriptsize $\pm$ 0.002} & 0.032 {\scriptsize $\pm$ 0.002} \\
Joint & 0.010 {\scriptsize $\pm$ 0.001} & \textbf{0.014} {\scriptsize $\pm$ 0.002} & \textbf{0.014} {\scriptsize $\pm$ 0.002} & 0.012 {\scriptsize $\pm$ 0.001} \\
\midrule
\multicolumn{5}{l}{\textbf{Worst-topic gASCE} $\downarrow$} \\
Uncalibrated & 0.382 {\scriptsize $\pm$ 0.053} & 0.534 {\scriptsize $\pm$ 0.047} & 0.328 {\scriptsize $\pm$ 0.054} & 0.748 {\scriptsize $\pm$ 0.050} \\
Marginal HB & 0.061 {\scriptsize $\pm$ 0.020} & 0.133 {\scriptsize $\pm$ 0.027} & 0.084 {\scriptsize $\pm$ 0.024} & 0.436 {\scriptsize $\pm$ 0.050} \\
Topic-only & 0.031 {\scriptsize $\pm$ 0.017} & 0.066 {\scriptsize $\pm$ 0.021} & \textbf{0.046} {\scriptsize $\pm$ 0.018} & \textbf{0.090} {\scriptsize $\pm$ 0.031} \\
CoT-only & \textbf{0.029} {\scriptsize $\pm$ 0.013} & 0.092 {\scriptsize $\pm$ 0.022} & 0.057 {\scriptsize $\pm$ 0.019} & 0.192 {\scriptsize $\pm$ 0.042} \\
Joint & 0.030 {\scriptsize $\pm$ 0.015} & \textbf{0.058} {\scriptsize $\pm$ 0.017} & 0.066 {\scriptsize $\pm$ 0.018} & 0.096 {\scriptsize $\pm$ 0.030} \\
\bottomrule
\end{tabular}
\endgroup

    \caption{Mean and worst-topic calibration error (gASCE) on held-out examples across four datasets. Lower values indicate better topic-level calibration.}
    \label{tab:cot-topic-calibration}
    
\end{table*}

\begin{figure*}[htbp]
    \centering
    \includegraphics[width=\textwidth]{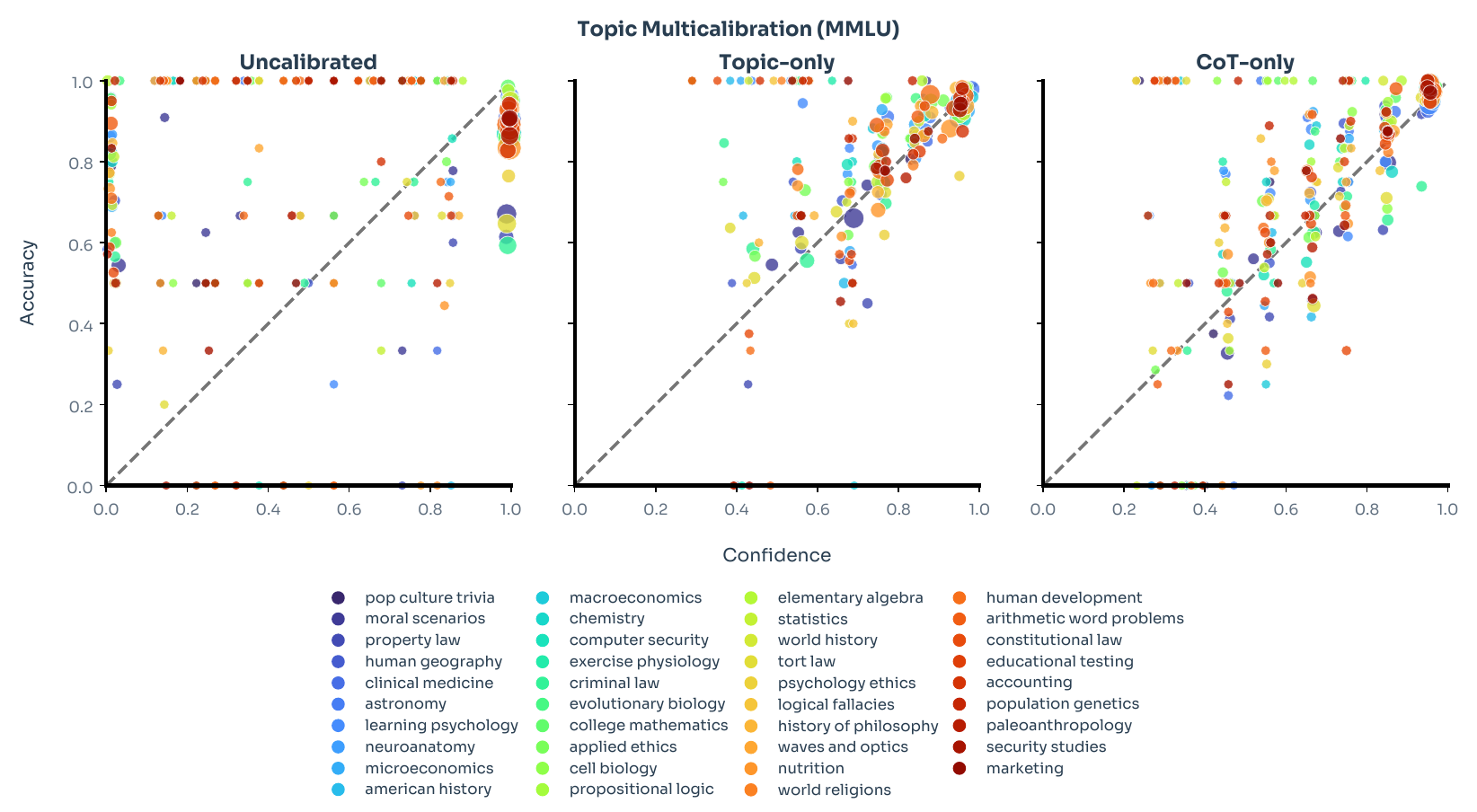}
    \caption{Topic-level reliability for all 39 MMLU topics, before calibration
    and after topic-only or CoT-only IGHB. Each point shows mean confidence and
    accuracy in a nonempty topic--confidence bin, using ten equal-width bins.
    Color identifies the topic and point size indicates sample count.
    The predictors and evaluation examples are the same as in
    Figure~\ref{fig:topic-transfer-reliability}.}
    \label{fig:topic-transfer-all-topics}
\end{figure*}

\begin{figure*}[htbp]
    \centering
    \includegraphics[width=\textwidth]{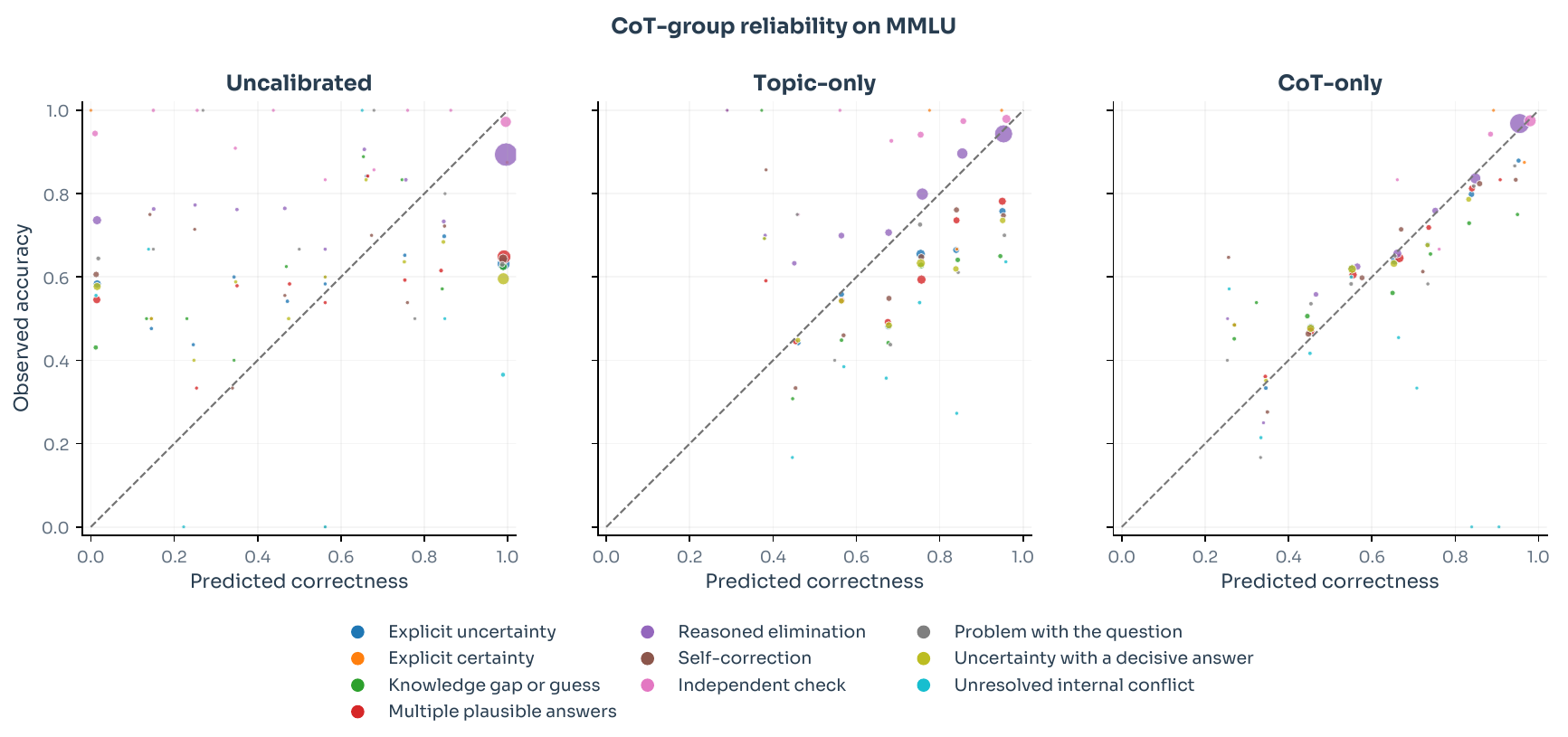}
    \caption{CoT-group reliability on held-out MMLU responses, before calibration
    and after topic-only or CoT-only IGHB. This reverses the audit in
    Figure~\ref{fig:topic-transfer-all-topics}: predictions are evaluated within
    CoT groups rather than topic groups. Each color denotes the Yes-group of
    one diagnostic question.}
    \label{fig:cot-transfer-all-groups}
\end{figure*}

\begin{figure*}[htbp]
    \centering
    \includegraphics[width=0.9\textwidth]{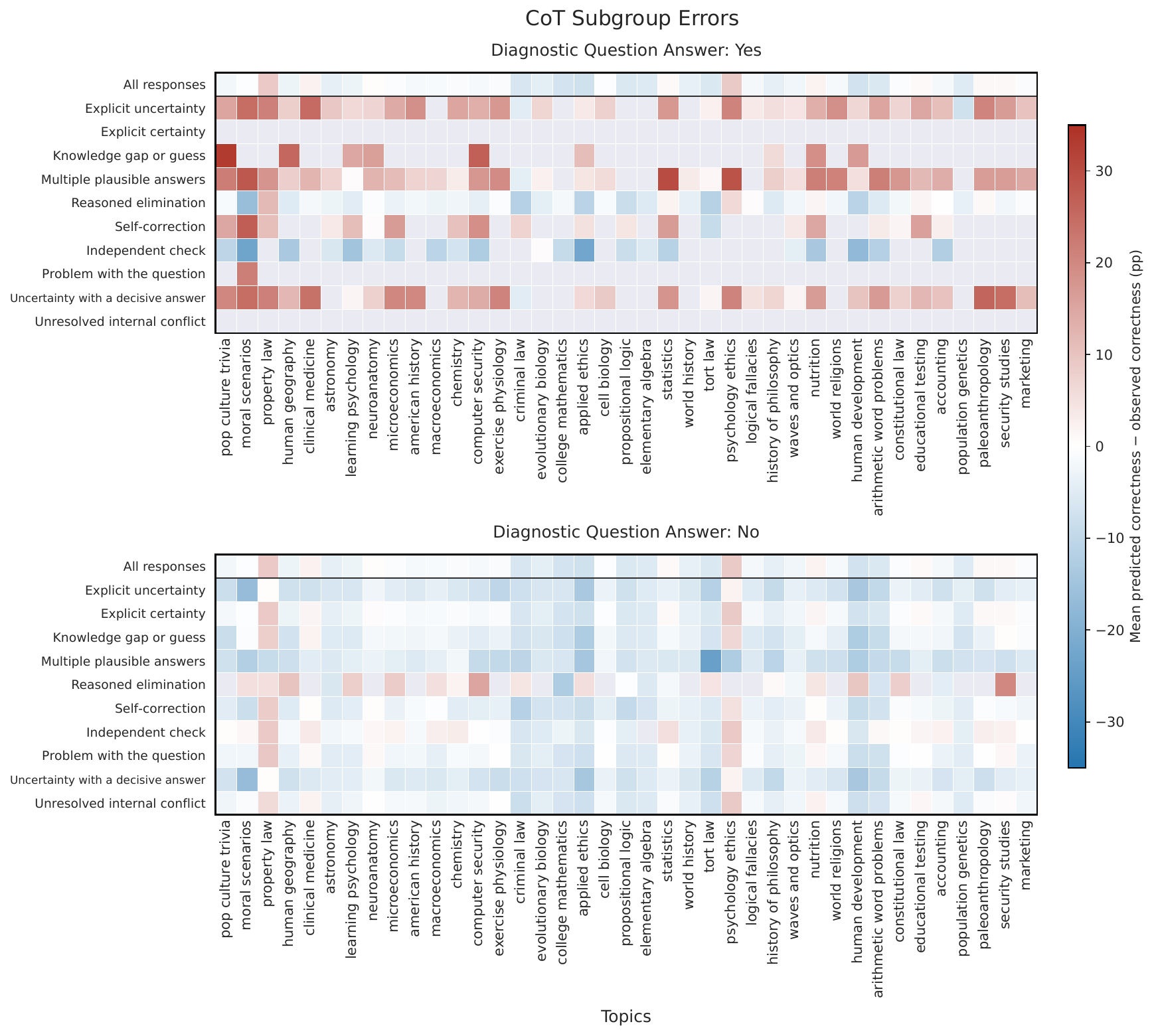}
    \caption{CoT subgroup errors after topic-only IGHB calibration on held-out MMLU responses from Qwen3-8B. Columns show all 39 embedding-derived topic groups, using the same names as Figure~\ref{fig:topic-subject-overlap}. Rows show diagnostic questions, with Yes answers in the top panel and No answers in the bottom panel. The pooled topic error is repeated in the first row. Colors show mean predicted correctness minus observed correctness, in percentage points (pp); positive values indicate overprediction. Gray cells contain fewer than 20 responses.}
    \label{fig:heatmap}
\end{figure*}


\subsection{Model and Algorithm Ablations}
\label{app:ablations}

\xhdr{Diagnostic labeler.}
We repeat the case study of \Cref{sec:casestudy} with three models answering the ten diagnostic questions in \Cref{tab:cot_questions}, using the labeling protocol in \Cref{app:cot-groupings}. In the main body, we use Gemma4-31B for two main reasons: (i) we expected a larger model to label these properties more accurately, and (ii) it belongs to a different model family from Qwen3-8B, which generated the traces. Here we test whether Qwen3-8B can label its own traces and whether the smaller Qwen3-4B can supply useful labels. For each labeler, we refit CoT-only and joint IGHB using the same traces, initial scores, correctness labels, data splits, topic groups, and validation procedure. All conditions are evaluated on the same held-out questions.
Joint IGHB lowers held-out Brier score relative to topic-only IGHB on all four datasets with each labeler (\Cref{tab:diagnostic-labeler}). Qwen3-4B gives smaller gains than Gemma4-31B, and its CoT-only fits do not consistently beat topic-only. Adding its groups to the topic groups still improves held-out Brier score on every dataset. 

\begin{table*}[htbp]
    \centering    \begingroup
\small
\setlength{\tabcolsep}{3pt}
\renewcommand{\arraystretch}{1.15}
\begin{tabular}{@{}llrrrr@{}}
\toprule
Calibration & Labeler & MMLU & MMLU-Pro & TriviaQA & PopQA \\
\midrule
\addlinespace[3pt]
\multicolumn{6}{l}{\textbf{Brier score} $\downarrow$} \\
\addlinespace[3pt]
Topic-only & --- & 0.127 {\scriptsize $\pm$ 0.003} & 0.163 {\scriptsize $\pm$ 0.004} & 0.185 {\scriptsize $\pm$ 0.004} & 0.114 {\scriptsize $\pm$ 0.004} \\
CoT-only & Qwen3-4B & 0.127 {\scriptsize $\pm$ 0.003} & 0.163 {\scriptsize $\pm$ 0.004} & 0.183 {\scriptsize $\pm$ 0.004} & 0.129 {\scriptsize $\pm$ 0.004} \\
CoT-only & Qwen3-8B & 0.123 {\scriptsize $\pm$ 0.003} & 0.163 {\scriptsize $\pm$ 0.004} & 0.174 {\scriptsize $\pm$ 0.004} & 0.126 {\scriptsize $\pm$ 0.004} \\
CoT-only & Gemma4-31B & 0.113 {\scriptsize $\pm$ 0.003} & 0.141 {\scriptsize $\pm$ 0.004} & 0.164 {\scriptsize $\pm$ 0.004} & 0.125 {\scriptsize $\pm$ 0.004} \\
Joint & Qwen3-4B & 0.124 {\scriptsize $\pm$ 0.003} & 0.155 {\scriptsize $\pm$ 0.004} & 0.178 {\scriptsize $\pm$ 0.004} & 0.109 {\scriptsize $\pm$ 0.004} \\
Joint & Qwen3-8B & 0.121 {\scriptsize $\pm$ 0.003} & 0.149 {\scriptsize $\pm$ 0.004} & 0.170 {\scriptsize $\pm$ 0.004} & \textbf{0.108} {\scriptsize $\pm$ 0.003} \\
Joint & Gemma4-31B & \textbf{0.112} {\scriptsize $\pm$ 0.003} & \textbf{0.139} {\scriptsize $\pm$ 0.004} & \textbf{0.163} {\scriptsize $\pm$ 0.004} & 0.109 {\scriptsize $\pm$ 0.004} \\
\midrule
\addlinespace[3pt]
\multicolumn{6}{l}{\textbf{Mean topic gASCE} $\downarrow$} \\
\addlinespace[3pt]
Topic-only & --- & \textbf{0.008} {\scriptsize $\pm$ 0.001} & 0.015 {\scriptsize $\pm$ 0.002} & 0.014 {\scriptsize $\pm$ 0.002} & 0.009 {\scriptsize $\pm$ 0.002} \\
CoT-only & Qwen3-4B & 0.013 {\scriptsize $\pm$ 0.001} & 0.022 {\scriptsize $\pm$ 0.002} & 0.020 {\scriptsize $\pm$ 0.002} & 0.031 {\scriptsize $\pm$ 0.003} \\
CoT-only & Qwen3-8B & 0.013 {\scriptsize $\pm$ 0.001} & 0.030 {\scriptsize $\pm$ 0.002} & 0.021 {\scriptsize $\pm$ 0.002} & 0.030 {\scriptsize $\pm$ 0.002} \\
CoT-only & Gemma4-31B & 0.010 {\scriptsize $\pm$ 0.001} & 0.018 {\scriptsize $\pm$ 0.002} & 0.017 {\scriptsize $\pm$ 0.002} & 0.032 {\scriptsize $\pm$ 0.002} \\
Joint & Qwen3-4B & 0.010 {\scriptsize $\pm$ 0.001} & 0.015 {\scriptsize $\pm$ 0.002} & 0.016 {\scriptsize $\pm$ 0.002} & 0.011 {\scriptsize $\pm$ 0.001} \\
Joint & Qwen3-8B & 0.010 {\scriptsize $\pm$ 0.001} & 0.018 {\scriptsize $\pm$ 0.002} & 0.015 {\scriptsize $\pm$ 0.002} & \textbf{0.009} {\scriptsize $\pm$ 0.001} \\
Joint & Gemma4-31B & 0.010 {\scriptsize $\pm$ 0.001} & \textbf{0.014} {\scriptsize $\pm$ 0.002} & \textbf{0.014} {\scriptsize $\pm$ 0.002} & 0.012 {\scriptsize $\pm$ 0.001} \\
\midrule
\addlinespace[3pt]
\multicolumn{6}{l}{\textbf{Worst-topic gASCE} $\downarrow$} \\
\addlinespace[3pt]
Topic-only & --- & 0.031 {\scriptsize $\pm$ 0.017} & 0.066 {\scriptsize $\pm$ 0.021} & \textbf{0.046} {\scriptsize $\pm$ 0.018} & 0.090 {\scriptsize $\pm$ 0.031} \\
CoT-only & Qwen3-4B & 0.060 {\scriptsize $\pm$ 0.023} & 0.077 {\scriptsize $\pm$ 0.020} & 0.065 {\scriptsize $\pm$ 0.020} & 0.264 {\scriptsize $\pm$ 0.056} \\
CoT-only & Qwen3-8B & 0.048 {\scriptsize $\pm$ 0.019} & 0.158 {\scriptsize $\pm$ 0.030} & 0.063 {\scriptsize $\pm$ 0.019} & 0.237 {\scriptsize $\pm$ 0.052} \\
CoT-only & Gemma4-31B & 0.029 {\scriptsize $\pm$ 0.013} & 0.092 {\scriptsize $\pm$ 0.022} & 0.057 {\scriptsize $\pm$ 0.019} & 0.192 {\scriptsize $\pm$ 0.042} \\
Joint & Qwen3-4B & \textbf{0.022} {\scriptsize $\pm$ 0.013} & \textbf{0.047} {\scriptsize $\pm$ 0.016} & 0.067 {\scriptsize $\pm$ 0.023} & 0.094 {\scriptsize $\pm$ 0.031} \\
Joint & Qwen3-8B & 0.023 {\scriptsize $\pm$ 0.013} & 0.088 {\scriptsize $\pm$ 0.024} & 0.048 {\scriptsize $\pm$ 0.018} & \textbf{0.059} {\scriptsize $\pm$ 0.025} \\
Joint & Gemma4-31B & 0.030 {\scriptsize $\pm$ 0.015} & 0.058 {\scriptsize $\pm$ 0.017} & 0.066 {\scriptsize $\pm$ 0.018} & 0.096 {\scriptsize $\pm$ 0.030} \\
\bottomrule
\end{tabular}
\endgroup

    \caption{Results across diagnostic labelers, audited over topics. Other baselines appear in \Cref{tab:cot-prediction-performance}.}
    \label{tab:diagnostic-labeler}
\end{table*}

\xhdr{Calibration algorithm.}
To test whether the Brier gain from CoT groups depends on the calibration algorithm, we repeat the topic-only, CoT-only, and joint fits of \Cref{sec:casestudy} with iterative grouped linear binning (IGLB). IGLB and the main IGHB fits use the same responses, initial scores, correctness labels, groups, and splits. The main difference is that IGLB fits a linear correction within each selected region, while IGHB applies a constant shift.

Joint IGLB lowers held-out Brier score relative to topic-only IGLB on all four datasets (\Cref{tab:iglb-performance}). The mean and worst-topic calibration errors change in both directions. Thus, the Brier gain from adding CoT groups also appears with IGLB on these benchmarks.

\begin{table*}[htbp]
    \centering
    \begingroup
\small
\setlength{\tabcolsep}{4pt}
\renewcommand{\arraystretch}{1.15}
\begin{tabular}{@{}lrrrr@{}}
\toprule
Calibration & MMLU & MMLU-Pro & TriviaQA & PopQA \\
\midrule
\multicolumn{5}{l}{\textbf{Brier score} $\downarrow$} \\
Topic-only & 0.127 {\scriptsize $\pm$ 0.003} & 0.161 {\scriptsize $\pm$ 0.004} & 0.181 {\scriptsize $\pm$ 0.004} & 0.109 {\scriptsize $\pm$ 0.004} \\
CoT-only & 0.112 {\scriptsize $\pm$ 0.003} & 0.141 {\scriptsize $\pm$ 0.004} & 0.159 {\scriptsize $\pm$ 0.004} & 0.119 {\scriptsize $\pm$ 0.004} \\
Joint & 0.112 {\scriptsize $\pm$ 0.003} & 0.141 {\scriptsize $\pm$ 0.004} & 0.158 {\scriptsize $\pm$ 0.004} & 0.103 {\scriptsize $\pm$ 0.004} \\
\midrule
\multicolumn{5}{l}{\textbf{Mean topic gASCE} $\downarrow$} \\
Topic-only & 0.009 {\scriptsize $\pm$ 0.001} & 0.015 {\scriptsize $\pm$ 0.002} & 0.018 {\scriptsize $\pm$ 0.002} & 0.009 {\scriptsize $\pm$ 0.001} \\
CoT-only & 0.010 {\scriptsize $\pm$ 0.001} & 0.017 {\scriptsize $\pm$ 0.002} & 0.015 {\scriptsize $\pm$ 0.002} & 0.026 {\scriptsize $\pm$ 0.002} \\
Joint & 0.012 {\scriptsize $\pm$ 0.001} & 0.014 {\scriptsize $\pm$ 0.002} & 0.016 {\scriptsize $\pm$ 0.002} & 0.010 {\scriptsize $\pm$ 0.001} \\
\midrule
\multicolumn{5}{l}{\textbf{Worst-topic gASCE} $\downarrow$} \\
Topic-only & 0.042 {\scriptsize $\pm$ 0.015} & 0.051 {\scriptsize $\pm$ 0.016} & 0.063 {\scriptsize $\pm$ 0.019} & 0.062 {\scriptsize $\pm$ 0.022} \\
CoT-only & 0.023 {\scriptsize $\pm$ 0.011} & 0.071 {\scriptsize $\pm$ 0.018} & 0.049 {\scriptsize $\pm$ 0.017} & 0.165 {\scriptsize $\pm$ 0.041} \\
Joint & 0.043 {\scriptsize $\pm$ 0.016} & 0.071 {\scriptsize $\pm$ 0.017} & 0.052 {\scriptsize $\pm$ 0.016} & 0.081 {\scriptsize $\pm$ 0.030} \\
\bottomrule
\end{tabular}
\endgroup

    \caption{IGLB calibration on the held-out examples. Topic gASCE uses the same eligible topics and ten confidence bins as \Cref{tab:cot-topic-calibration}. Lower is better.}
    \label{tab:iglb-performance}
\end{table*}

\subsection{Cross-Benchmark Calibration Transfer}
\label{app:cross-benchmark-transfer}

At deployment, a reliability predictor may encounter questions whose distribution differs from the data used to calibrate it. Obtaining correctness labels and refitting the predictor for each new setting can be costly. We therefore ask to what extent calibration transfers between benchmarks without refitting on target labels. To do this, we evaluate all twelve possible combinations of transfers among MMLU, MMLU-Pro, TriviaQA, and PopQA (e.g., we might want to calibrate over MMLU but then evaluate over MMLU-Pro). Throughout this experiment, we keep the answer-generating model, Qwen3-8B, and the True/False confidence score fixed. For each direction, one benchmark supplies the data for group construction, calibration, and validation. We then freeze the resulting predictor and evaluate it on the target benchmark's held-out examples.

For each target question, we compute its embedding as in \Cref{app:topic}. We pass it through the PCA projection and Gaussian mixture fitted on the source benchmark and obtain a probability for each source topic. We then assign it to each topic whose probability exceeds the threshold set on the source data. CoT memberships use the same labeling procedure in \Cref{app:cot-groupings}. The source calibrator then uses these groups to update the initial confidence score.

\begin{table*}[htpb]
    \centering
\begingroup\small
\setlength{\tabcolsep}{7pt}
\renewcommand{\arraystretch}{1.12}
\begin{tabular}{@{}lrrrrr@{}}
\toprule
Calibrated on & Raw & HB & Topic-only & CoT-only & Joint \\
\midrule
\multicolumn{6}{@{}l}{\textbf{Evaluated on MMLU}} \\
\textbf{MMLU} & 0.2329 & 0.1366 & 0.1267 & 0.1127 & 0.1124 \\
MMLU-Pro & 0.2329 & 0.1418 & 0.1416 & 0.1177 & 0.1187 \\
TriviaQA & 0.2329 & 0.1950 & 0.1854 & 0.1284 & 0.1303 \\
PopQA & 0.2329 & 0.4385 & 0.4322 & 0.2399 & 0.2688 \\
\midrule
\multicolumn{6}{@{}l}{\textbf{Evaluated on MMLU-Pro}} \\
MMLU & 0.3591 & 0.1872 & 0.1714 & 0.1473 & 0.1465 \\
\textbf{MMLU-Pro} & 0.3591 & 0.1839 & 0.1633 & 0.1413 & 0.1388 \\
TriviaQA & 0.3591 & 0.2553 & 0.2459 & 0.1637 & 0.1684 \\
PopQA & 0.3591 & 0.4613 & 0.4586 & 0.3231 & 0.3371 \\
\midrule
\multicolumn{6}{@{}l}{\textbf{Evaluated on TriviaQA}} \\
MMLU & 0.2940 & 0.2590 & 0.2721 & 0.1835 & 0.1896 \\
MMLU-Pro & 0.2940 & 0.2466 & 0.2449 & 0.1736 & 0.1758 \\
\textbf{TriviaQA} & 0.2940 & 0.2090 & 0.1855 & 0.1644 & 0.1625 \\
PopQA & 0.2940 & 0.3189 & 0.2708 & 0.2082 & 0.2118 \\
\midrule
\multicolumn{6}{@{}l}{\textbf{Evaluated on PopQA}} \\
MMLU & 0.3822 & 0.4635 & 0.4680 & 0.2447 & 0.2591 \\
MMLU-Pro & 0.3822 & 0.4373 & 0.3655 & 0.2030 & 0.2042 \\
TriviaQA & 0.3822 & 0.2422 & 0.2195 & 0.1589 & 0.1693 \\
\textbf{PopQA} & 0.3822 & 0.1482 & 0.1142 & 0.1248 & 0.1090 \\
\bottomrule
\end{tabular}
\endgroup

    \caption{Brier score under benchmark transfer. Rows identify the calibration benchmark; bold labels mark fitting on the target's separate calibration split. Lower is better.}
    \label{tab:cross-benchmark-brier}
\end{table*}

\begin{table*}[htpb]
    \centering
\begingroup\small
\setlength{\tabcolsep}{7pt}
\renewcommand{\arraystretch}{1.12}
\begin{tabular}{@{}lrrrrr@{}}
\toprule
Calibrated on & Raw & HB & Topic-only & CoT-only & Joint \\
\midrule
\multicolumn{6}{@{}l}{\textbf{Evaluated on MMLU}} \\
\textbf{MMLU} & 0.106 & 0.011 & 0.008 & 0.010 & 0.010 \\
MMLU-Pro & 0.106 & 0.017 & 0.023 & 0.017 & 0.017 \\
TriviaQA & 0.106 & 0.074 & 0.066 & 0.026 & 0.026 \\
PopQA & 0.106 & 0.322 & 0.316 & 0.131 & 0.158 \\
\midrule
\multicolumn{6}{@{}l}{\textbf{Evaluated on MMLU-Pro}} \\
MMLU & 0.216 & 0.031 & 0.023 & 0.023 & 0.022 \\
\textbf{MMLU-Pro} & 0.216 & 0.027 & 0.015 & 0.018 & 0.014 \\
TriviaQA & 0.216 & 0.106 & 0.098 & 0.036 & 0.040 \\
PopQA & 0.216 & 0.306 & 0.312 & 0.188 & 0.198 \\
\midrule
\multicolumn{6}{@{}l}{\textbf{Evaluated on TriviaQA}} \\
MMLU & 0.115 & 0.066 & 0.094 & 0.027 & 0.033 \\
MMLU-Pro & 0.115 & 0.044 & 0.054 & 0.018 & 0.021 \\
\textbf{TriviaQA} & 0.115 & 0.022 & 0.014 & 0.017 & 0.014 \\
PopQA & 0.115 & 0.127 & 0.094 & 0.054 & 0.051 \\
\midrule
\multicolumn{6}{@{}l}{\textbf{Evaluated on PopQA}} \\
MMLU & 0.274 & 0.353 & 0.365 & 0.146 & 0.161 \\
MMLU-Pro & 0.274 & 0.323 & 0.259 & 0.106 & 0.106 \\
TriviaQA & 0.274 & 0.134 & 0.118 & 0.065 & 0.074 \\
\textbf{PopQA} & 0.274 & 0.041 & 0.009 & 0.032 & 0.012 \\
\bottomrule
\end{tabular}
\endgroup

    \caption{Membership-weighted mean gASCE over target-defined topics. Rows identify the calibration benchmark; bold labels mark fitting on the target's separate calibration split. Lower is better.}
    \label{tab:cross-benchmark-calibration}
\end{table*}

\Cref{tab:cross-benchmark-brier,tab:cross-benchmark-calibration} show that joint transfer improves Brier score and mean topic calibration error over topic-only transfer in all twelve directions. CoT-only usually transfers better than joint, with lower Brier in eleven directions and lower calibration error in ten. Joint beats raw scores on both metrics in eleven directions. Topic-only calibration transfers less well than CoT-only calibration. This may be because source topics do not represent target questions well, whereas the diagnostic definitions remain unchanged. As expected, joint calibration fitted on the evaluation benchmark has lower Brier than joint calibration transferred from another benchmark in every direction.

\section{Extension to Safety Task}
\label{app:safety}

We extend the application of reasoning-aware multicalibration to BeaverTails \citep{ji2023beavertailsimprovedsafetyalignment}, an established safety benchmark. Qwen3-8B reasons about the safety of a stored response given its user request; the target is the human unsafe-response label. Then, a separate call provides the initial probability through the True/False token score. We use \texttt{330k\_train} split, remove request–response pairs with inconsistent safety labels, and retain one response per distinct request using a reproducible selection rule that is independent of the label. The resulting dataset contains 15,461 examples, of which 4,680 are used for evaluation.

Similarly to the protocol in \Cref{sec:casestudy}, we use UAE request-embedding topic groups and multicalibrate via IGHB. We use Gemma4-31B to label the original ten diagnostics and also add five additional questions about safety reasoning (\Cref{tab:safety-questions}). The labeler sees only the CoT of the model being monitored. Scores, targets, examples, splits, and topic groups are fixed across conditions. 
Joint calibration with fifteen diagnostics has the lowest Brier score (\Cref{tab:safety-performance}). Mean topic calibration error changes little, while worst-topic error increases. CoT-only calibration improves with the added questions but remains close to topic-only Brier performance. 

\begin{table}[htbp]
    \centering
    \begingroup\small
\setlength{\tabcolsep}{6pt}
\renewcommand{\arraystretch}{1.15}
\begin{tabular}{@{}lrrr@{}}
\toprule
Calibration & Brier & Mean topic gASCE & Worst-topic gASCE \\
\midrule
Uncalibrated & 0.166 {\scriptsize $\pm$ 0.005} & 0.047 {\scriptsize $\pm$ 0.003} & 0.552 {\scriptsize $\pm$ 0.056} \\
Topic-only & 0.107 {\scriptsize $\pm$ 0.003} & 0.011 {\scriptsize $\pm$ 0.001} & 0.068 {\scriptsize $\pm$ 0.017} \\
CoT-only (10) & 0.118 {\scriptsize $\pm$ 0.003} & 0.019 {\scriptsize $\pm$ 0.001} & 0.214 {\scriptsize $\pm$ 0.038} \\
Joint (10) & 0.106 {\scriptsize $\pm$ 0.003} & \textbf{0.010} {\scriptsize $\pm$ 0.001} & \textbf{0.062} {\scriptsize $\pm$ 0.017} \\
CoT-only (15) & 0.108 {\scriptsize $\pm$ 0.003} & 0.016 {\scriptsize $\pm$ 0.001} & 0.177 {\scriptsize $\pm$ 0.037} \\
Joint (15) & \textbf{0.102} {\scriptsize $\pm$ 0.003} & 0.010 {\scriptsize $\pm$ 0.001} & 0.106 {\scriptsize $\pm$ 0.023} \\
\bottomrule
\end{tabular}\endgroup

    \caption{Brier score and topic calibration for the safety prediction task, with ten or fifteen diagnostics, audited over topics. Values are estimates $\pm$ one bootstrap standard error (\Cref{sec:casestudy}). Bold marks the lowest point estimate in each column.}
    \label{tab:safety-performance}
\end{table}

\begin{table}[htbp]
    \centering
    \begingroup\small
\renewcommand{\arraystretch}{1.12}
\begin{tabular}{@{}p{0.19\linewidth}p{0.66\linewidth}r@{}}
\toprule
CoT property & Diagnostic question & Yes \\
\midrule
\raggedright Request and response & \raggedright Does the trace explicitly distinguish the user's request from the response being evaluated when explaining its safety judgment? & 4,399 \\
\addlinespace[4pt]
\raggedright Discussion and assistance & \raggedright Does the trace explicitly distinguish discussing or describing a harmful activity from providing practical help to carry it out? & 524 \\
\addlinespace[4pt]
\raggedright Harm mechanism & \raggedright Does the trace describe a concrete way that the response could cause or facilitate harm? Do not count a bare harm-category name or the word unsafe. & 1,250 \\
\addlinespace[4pt]
\raggedright Refusals and disclaimers & \raggedright Does the trace assess whether a refusal, warning, or disclaimer changes the safety of the rest of the response? Do not count merely noting that a refusal, warning, or disclaimer is present. & 183 \\
\addlinespace[4pt]
\raggedright Context dependence & \raggedright Does the trace explicitly make its safety judgment depend on an assumption or condition about intent, authorization, audience, or use context? Do not count merely mentioning context. & 352 \\
\addlinespace[4pt]
\bottomrule
\end{tabular}\endgroup

    \caption{Additional safety diagnostics questions. The ``Yes'' column reports the distribution of positive answers across the questions.}
    \label{tab:safety-questions}
\end{table}

\end{document}